\documentclass[12pt]{iopart}

\usepackage{iopams}
\expandafter\let\csname equation*\endcsname\relax
\expandafter\let\csname endequation*\endcsname\relax
\usepackage{amsthm}
\usepackage{amsmath}

\usepackage{booktabs}
\usepackage[justification=justified, margin=15pt, font={footnotesize}]{caption}

\expandafter\let\csname equation*\endcsname\relax
\expandafter\let\csname endequation*\endcsname\relax

\usepackage{graphicx}
\usepackage{flafter}  % no floats before the place of insertion

\usepackage{cite}
\usepackage[colorlinks,citecolor=blue,urlcolor=blue]{hyperref}
\usepackage{url}
\usepackage[caption=false]{subfig}
\usepackage[colorlinks,citecolor=blue,urlcolor=blue]{hyperref}

\newcommand{\Eq}[1]{(\ref{eq:#1})}
\newcommand{\Th}[1]{Th.~\ref{thm:#1}}

\newcommand{\Sec}[1]{\S \ref{sec:#1}}
\newcommand{\Fig}[1]{Fig.~\ref{fig:#1}}
\newcommand{\Tbl}[1]{Table~\ref{tbl:#1}}

\newcommand{\InsertFig}[5][h!t]{
\begin{figure}[#1]
       \centerline{
         \includegraphics[scale=1.2]{#2}
       }
       \caption{#3
       \label{fig:#4}}
\end{figure}}

\newcommand{\InsertFigTwo}[5] {
\begin{figure}[b]
       \centerline{
         \includegraphics[width=#5]{#1}
         \includegraphics[width=#5]{#2}
       }
       \caption{#3
       \label{fig:#4}}
\end{figure}}

\newcommand{\bC}{{\mathbb{ C}}}

\newcommand{\bR}{{\mathbb{ R}}}

\newcommand{\bT}{{\mathbb{ T}}}
\newcommand{\bZ}{{\mathbb{ Z}}}
\newcommand{\cA}{{\cal A}}

\newcommand{\cL}{{\cal L}}
\newcommand{\cO}{{\cal O}}

\newcommand{\eps}{\varepsilon}

\newcommand{\txtQ}{{\text{Q}}}

\newcommand{\sgn}{\mathop{\rm sgn}\nolimits}

\renewcommand{\tr}[1]{\mathop{\rm tr}({#1})}

\newtheorem{thm}{Theorem}

\newcommand{\beq}[1]{\begin{equation}\label{eq:#1}}
\newcommand{\eeq}{\end{equation}}

\newenvironment{se}[1]{\equation\label{eq:#1}\aligned}{\endaligned\endequation}
\newcommand{\bsplit}[1]{\begin{se}{#1}}
\newcommand{\esplit}{\end{se}}

\newenvironment{example}[1][]
  {
	\setlength \leftmargini {1.0em}		%controls the indentation for quote & list
	\setlength \topsep {0.5em}			%controls the top space for quote & list
	\begin{quote}
	{\it Example#1} }
	{\end{quote}
  }
\newcommand{\bexam}[1][:]{\begin{example}[#1]}
\newcommand{\eexam}{\end{example}}

\newcommand{\uv}{\text{v}}  % basis and dual basis vectors
\providecommand*{\oneD}{\textsc{1d}}
\providecommand*{\twoD}{\textsc{2d}}
\providecommand*{\threeD}{\textsc{3d}}
\providecommand*{\fourD}{\textsc{4d}}

\newcommand{\Hen} {H\'enon}
\providecommand*{\Areg}{A_{\text{reg}}}   % Area of bounded orbits
\providecommand*{\tesc}{t_{\text{esc}}}   % escape time.

\newcommand{\Cs}{C^\text{s}}
\newcommand{\Ca}{C^\text{a}}

\newcommand{\tq}{\tilde{q}}
\newcommand{\tp}{\tilde{p}}
\newcommand{\tK}{\tilde{K}}
\newcommand{\hq}{\hat{q}}
\newcommand{\hp}{\hat{p}}

\newcommand{\ahenon}{a_{\text{h}}}
\newcommand{\ahenOne}{a_{\text{h}1}}
\newcommand{\ahenTwo}{a_{\text{h}2}}
\newcommand{\hFro}{L}

\newcommand{\movierefall}{For a rotating view see
\href{http://www.comp-phys.tu-dresden.de/supp/}{http://www.comp-phys.tu-dresden.de/supp/}.}

\begin{document}
\title[The Quadfurcation in Moser's Map]{Elliptic Bubbles in Moser's 4D Quadratic Map:\\ the Quadfurcation}
\author[B\"acker and Meiss]{Arnd B\"acker$^{1,2}$, and James D.~Meiss$^{3}$}

\address{$^1$Technische Universit\"{a}t Dresden,\\
          Institut f\"{u}r Theoretische Physik and Center for Dynamics, 01062 Dresden, Germany}
\address{$^2$Max-Planck-Institut f\"{u}r Physik komplexer Systeme, N\"{o}thnitzer Strasse 38,\\
 01187 Dresden, Germany}
\address{$^3$University of Colorado,\\
Department of Applied Mathematics,
Boulder, CO 80309-0526, USA}

\begin{abstract}
Moser derived a normal form for the family of four-dimensional, quadratic, symplectic maps in 1994.
This six-parameter family generalizes H\'enon's ubiquitous \twoD{} map
and provides a local approximation for the dynamics of more general \fourD{} maps.
We show that the bounded dynamics of Moser's family is organized by a codimension-three
bifurcation that creates four fixed points---a bifurcation analogous to a doubled, saddle-center---which
we call a \textit{quadfurcation}.

In some sectors of parameter space a quadfurcation creates four fixed points from none,
and in others it is the collision of a pair of fixed points that re-emerge as two or possibly four.
In the simplest case the dynamics is similar to the cross product of a pair of H\'enon maps,
but more typically the stability of the created fixed points does not have this simple form.
Up to two of the fixed points can be doubly-elliptic and be surrounded by bubbles
of invariant two-tori; these dominate the set of bounded orbits. The quadfurcation can also
create one or two complex-unstable (Krein) fixed points.

Special cases of the quadfurcation correspond to a pair of weakly coupled H\'enon maps
near their saddle-center bifurcations.
The quadfurcation also occurs in the creation of accelerator modes in a \fourD{} standard map.

\end{abstract}
%\ams{34C37, 37C29, 37J45, 70H09}
%\pacs{02.40.-k, 05.45.-a, 45.20.Jj}
\vspace*{1ex}
\noindent

\vspace{2pc}

\date{\today}
\maketitle

%%%%%%%%%%%%%%%
%%  Introduction
%%%%%%%%%%%%%%%
\section{Introduction}
Multi-dimensional Hamiltonian systems model dynamics on scales ranging from zettameters, for the dynamics of stars in galaxies \cite{Cincotta02, Contopoulos13}, to nanometers, in atoms and molecules \cite{Gekle06, Paskauskas08}.
Hamiltonian flows generate symplectic maps on Poincar\'e sections \cite[\S9.14]{Meiss17a}, and numerical algorithms for these flows can be symplectic \cite{McLachlan06,Forest06}. Symplectic maps also arise directly in discrete-time models of such phenomena as molecular vibrations \cite{Gaspard89, Gillilan91},  stability of particle storage rings \cite{Warnock92,Dumas93}, heating of particles in plasmas \cite{Howard86}, microwave ionization of hydrogen \cite{Casati91} and chaos in celestial mechanics \cite{Wisdom91}

A map $f: \bR^{2n} \to \bR^{2n}$ is canonically symplectic for coordinates $x \in \bR^n$ and momenta $y \in \bR^n$ if its Jacobian matrix, $Df(x,y)$, satisfies
\beq{Symplectic}
	Df^T J Df = J, \quad J = \begin{pmatrix} 0 & -I \\ I & 0 \end{pmatrix},
\eeq
where $J$ is the Poisson matrix. In particular this implies that the map is volume preserving: $\det(Df) = 1$.

Perhaps the most famous symplectic map is the area-preserving map introduced by \Hen~in 1969 as an elemental model to inform his studies of celestial mechanics \cite{Henon69}. This map is also the simplest nonlinear symplectic map, since it contains a single quadratic term, and yet---as \Hen~showed---every quadratic area-preserving map can be reduced to his form \cite{Henon76}.

Quadratic maps are useful because they model the dynamics of smooth maps in the
neighborhood of a fixed point or periodic orbit.
For example, quadratic terms in the power series give a local description of the
dynamics near an accelerator mode of Chirikov's standard map \cite{Karney82}.
More generally, any symplectic diffeomorphism can be
$C^\infty$ approximated by a polynomial map on a compact set \cite{Turaev03}.

Higher-dimensional analogues of \Hen's map were proposed in \cite{Mao88}, and similar maps were used to study the stickiness of regions near an elliptic fixed point \cite{Ding90}, the resonant formation of periodic orbits and invariant circles \cite{Todesco94, Todesco96, Gemmi97, Vrahatis96}, bifurcations due to twist singularities \cite{Dullin03}, and the dynamics near a homoclinic orbit to a saddle-center fixed point \cite{Gonchenko04}. Such maps model a focussing-defocussing (FODO) magnet cell in a particle accelerator and have been used to study the structure of bounded orbits, the dynamic aperture, and robustness of invariant tori \cite{Bountis94, Vrahatis97, Giovannozzi98, Bountis06}.

In 1994, Moser \cite{Moser94} showed that every quadratic symplectic map
on $\bR^{2n}$ is conjugate to the form
\beq{AffineShear}
	f = \alpha \circ \sigma \circ \beta.
\eeq
Here $\alpha, \beta : \bR^{2n} \to \bR^{2n}$ are symplectic maps, $\beta$ is linear and $\alpha$ is affine, and $\sigma: \bR^n \times \bR^n \to \ \bR^n \times \bR^n$ is a symplectic shear:
\beq{Shear}
	\sigma(x,y) = (x, y - \nabla V(x)) ,
\eeq
where $V: \bR^n \to \bR$ is a cubic potential. There are several immediate consequences of this representation. Firstly, if the quadratic map $f$ has finitely many fixed points, as it generically will, then there are at most $2^n$ \cite{Moser94}.
Note that more generally a quadratic non-symplectic map on a $2n$-dimensional space could have as many as
$2^{2n}$ isolated fixed points. Secondly, since the inverse of $\sigma$ is also a quadratic shear of the same
form (replace $V$ by $-V$), the inverse of any quadratic, symplectic map is also quadratic. More generally,
the inverse of a quadratic diffeomorphism could be a polynomial map of higher degree \cite[Thm 1.5]
{BasConWri1982}; for example, the inverse of the volume-preserving map $(x,y,z,w) \mapsto (x,y+x^2,z+y^2,w+z^2)$
has degree eight. The form \Eq{AffineShear} also applies to cubic maps, but not to higher degree polynomial maps
\cite{Koch14}.

In this paper we study the dynamics of Moser's map in four dimensions. The normal form for the \fourD{} case is reviewed and slightly transformed for convenience in \Sec{MosersMap}. We argue in \Sec{Quadfurcation} that its fixed points are most properly viewed as arising from a bifurcation in which they emerge from a single fixed point as parameters are varied away from a codimension-three surface. Since this bifurcation often results in the creation of four fixed points, we call it a \textit{quadfurcation}, with thanks to Strogatz who, ``with tongue in cheek," proposed the term in an exercise for \oneD{} \textsc{ode}s in his well-known textbook \cite[Ex. 3.4.12]{Strogatz15}.

As we will see in \Sec{FixedPoints}, the unfolding of the quadfurcation in Moser's map can lead to
(i) the creation of four fixed points from none, or
(ii) a collision and re-emergence of two pairs of fixed points, or
(iii) even the collision of a pair leading to four fixed points.
The stability of these fixed points is investigated in \Sec{Stability}.
The unfolding of the quadfurcation along paths in parameter space is studied in \Sec{QuadLines}-\Sec{TwoToFour}.
When the map is reversible, \Sec{Symmetric}, additional cases occur including the simplest one: the
Cartesian product of a pair of area-preserving maps.
We investigate the creation of families of invariant two-tori,
as expected from KAM theory, around doubly elliptic fixed points in \Sec{Bubbles}.
In \Sec{Bounded} we observe that these bubbles of elliptic orbits strongly correlate with the regions of bounded dynamics.

Since Moser's map is affinely conjugate to the general quadratic, symplectic map, it must have a limit in which it reduces to a pair of uncoupled
\Hen~maps---we show this in \Sec{HenonMaps}. Finally, we show in \Sec{Accelerator} that the dynamics near an accelerator mode of the \fourD{} standard map (Froeschl\'e's map \cite{Fro1971}), can be modeled by Moser's map, and indeed, the local dynamics reduces to a coupled version of the \Hen~maps obtained in the previous section.

%%%%%%%%%%%%%%%%%%%%%%%%%%%%%%%%%%%%%%%%%%%%%%%%%%%%%%%%%%%%%%%%%%%%%%%%%%%%%
\section{Moser's quadratic, symplectic map}\label{sec:MosersMap}
%%%%%%%%%%%%%%%%%%%%%%%%%%%%%%%%%%%%%%%%%%%%%%%%%%%%%%%%%%%%%%%%%%%%%%%%%%%%%

%%%%%%%%%%%%%%%%%
%% Normal Form
%%%%%%%%%%%%%%%%%
\subsection{Four-Dimensional Normal Form}\label{sec:NormalForm}

For the two-dimensional case, the map \Eq{AffineShear} can be transformed by an affine coordinate change to the \Hen~map $H$,
\beq{HenonMap}
	H(x,y) =  (-y + \ahenon + x^2, x),
\eeq
with a single parameter $\ahenon$. When this map has an elliptic fixed point (for $-3 < \ahenon < 1$),
 it is conjugate to the map whose dynamics were first studied by \Hen~\cite{Henon69}. By a similar transformation Moser showed \cite{Moser94} that in four dimensions, \Eq{AffineShear} can generically be written as
\beq{MoserMap}
	(x', y') = f(x,y) = (C^{-T} (-y + \nabla V (x)), Cx) ,
\eeq
where
\bsplit{4DMap}
      C &= \begin{pmatrix} \alpha & \beta \\ \gamma & \delta \end{pmatrix} ,
      \quad  \\
      V &= A_1 x_1 + A_2 x_2 + \tfrac12 A_3 x_1^2 + \eps_2 x_1^3 + x_1 x_2^2 ,
\esplit
and $x\equiv(x_1, x_2) \in \bR^2$, $y\equiv(y_1, y_2) \in \bR^2$.
Here  there are two discrete parameters, $\eps_1 \equiv \det(C) = \alpha \delta -\beta\gamma\  = \pm 1$, and $\eps_2 \equiv \pm1$ or $0$.
The remaining six parameters are free.
It is convenient to think of the six real parameters of $f$ as $(A_1,A_2,A_3)$ and $(\alpha,\delta, \mu)$, where
\beq{muDefine}
	\mu = \beta + \gamma.
\eeq
Indeed, given these six, and the sign $\eps_1$, we can determine the off-diagonal elements of  $C$ from
\bsplit{betagamma}
	\beta, \gamma &= \tfrac12 \mu \pm \sqrt{ \eps_1 -\alpha\delta +\mu^2/4}.
\esplit
The choice of the sign here is unimportant since this simply replaces $C$ with $C^T$, and the resulting map is conjugate to the inverse of \Eq{MoserMap}; see \Sec{SecondDifference}.
Note that \Eq{betagamma} has real solutions only when $\mu^2 \ge 4(\alpha\delta - \eps_1)$,
and that $C$ is symmetric only at the lower bound of this inequality.

The map \Eq{MoserMap} is easily seen to be symplectic \Eq{Symplectic} as it is the composition of the symplectic shear \Eq{Shear},
the Poisson map, $J(x,y) = (-y,x)$, and the linear symplectic map
$$
	(x,y) \mapsto (C^{-T} x, Cy).
$$

%%%%%%%%%%%%%%%%%
%% Shifting to the Quadfurcation Point
%%%%%%%%%%%%%%%%%
\subsection{Shifted Coordinates}\label{sec:Shifted}

As a first step in the analysis of the dynamics of \Eq{MoserMap}, we will study its fixed points. To do this it is convenient to define shifted variables and parameters. There is a codimension-three set of parameters where the map has exactly one fixed point, and focussing on this set simplifies the calculations more generally.

For any matrix $C$ and when $\eps_2 \neq 0$, the map \Eq{MoserMap} has exactly one fixed point at
\bsplit{QuadrupleFixed}
  x^\txtQ &= (\delta, \tfrac12 \mu),\\
  y^\txtQ &= C x^\txtQ,
\esplit
when the parameters $(A_1,A_2,A_3)$ of the potential \Eq{4DMap} take the values
\bsplit{AStar}
   A_1^\txtQ &= 3\delta^2\eps_2 + \tfrac14\mu^2 \\
   A_2^\txtQ &=\delta\mu , \\
   A_3^\txtQ &=2 \alpha -6 \delta \eps_2.
\esplit
To see this, and to simplify the computations it is convenient to shift coordinates so that the origin is at the point $(x^\txtQ, y^\txtQ)$ and to define new shifted parameters:
\bsplit{ShiftedCoords}
	(\xi,\eta) &= (x-x^\txtQ,y-y^\txtQ) ,\\
	a &= A_1 - A_1^\txtQ + \delta(A_3-A_3^\txtQ) ,\\
	b &= A_2- A_2^\txtQ ,\\
	c &= A_3- A_3^\txtQ  .
\esplit
In these new coordinates, \Eq{MoserMap} becomes
\beq{ShiftedMap}
   (\xi', \eta') = M(\xi,\eta ) = (\xi + C^{-T}(-\eta + C\xi + \nabla U(\xi)) , C\xi ),
\eeq
where the new potential,
\beq{UPotential}
      U = a\xi_1 + b\xi_2 + \tfrac12 c \xi_1^2+ \eps_2 \xi_1^3 + \xi_1 \xi_2^2 ,
\eeq
is the same as $V$ from \Eq{4DMap} upon replacing $(A_1, A_2, A_3)$ by $(a,b,c)$.
This shifted form of Moser's quadratic, symplectic map
is convenient because several computations can be carried out more easily
and many of the expressions we obtain below will be more compact.

The map  \Eq{ShiftedMap} is generated by the discrete Lagrangian
\beq{Lagrangian}
	\cL(\xi,\xi') =  (\xi'-\xi)^T C \xi  - U(\xi) ,
\eeq
through the equation
$$
	\eta'd\xi' - \eta d\xi = d\cL(\xi,\xi') .
$$
In other words the map is implicitly defined by $\eta = -\partial_\xi \cL (\xi,\xi')$ and $\eta' = \partial_{\xi'} \cL(\xi,\xi')$. This means that $M$ is \textit{exact} symplectic \cite{Meiss15b}, and of course, that it preserves the symplectic form $d\xi \wedge d\eta$.
Note also that if we denote an orbit of \Eq{ShiftedMap} as a sequence
\beq{Orbit}
	\{(\xi_t,\eta_t) \in \bR^4 \,|\, (\xi_t,\eta_t) = M(\xi_{t-1},\eta_{t-1}),\, t \in \bZ\}
\eeq
and define the action of a finite portion by
$$
	\cA = \sum_{t=j}^{k-1} \cL(\xi_t,\xi_{t+1}),
$$
then each stationary point of $\cA$, for fixed endpoints,
is a segment of an orbit with the momentum determined by $\eta_{t+1} = C \xi_t$.

%%%%%%%%%%%%%%%%%
%% Second Difference Forem
%%%%%%%%%%%%%%%%%
\subsection{Second Difference Form and ODE Limit}\label{sec:SecondDifference}
The shifted form \Eq{ShiftedMap} of Moser's map can be written as a second difference equation. Denoting an orbit as \Eq{Orbit},
then $\eta_t = C \xi_{t-1}$ and the map \Eq{ShiftedMap} is equivalent to
\beq{SecondDiff1}
	C^T (\xi_{t+1}-\xi_t) - C(\xi_t - \xi_{t-1}) = \nabla U(\xi_t) .
\eeq
One immediate consequence of \Eq{SecondDiff1} is that the replacement $C \to C^T$ is clearly equivalent to inverting the map. Therefore the invariant sets of the Moser map with $C \to C^T$ are the same as those of the original map. Similarly, note that the replacement $C\to -C$ together with $\xi \to -\xi$ and $c \to -c$ leaves the Moser map invariant.
We will also use the form \Eq{SecondDiff1} in \Sec{Stability} and \Sec{Bounded}.

To emphasize the different roles of the symmetric and antisymmetric parts of $C$, let
\bsplit{CSym}
	C &= \Cs + \Ca, \\
	\Cs &\equiv \tfrac12(C+C^T) =\begin{pmatrix} \alpha & \mu/2 \\ \mu/2 & \delta \end{pmatrix},\\
	\Ca &\equiv \tfrac1{2}(C-C^T) = \begin{pmatrix} 0 & \nu/2 \\ -\nu/2 & 0 \end{pmatrix},
\esplit
where $\nu = \beta-\gamma$ and, as before $\mu = \beta + \gamma$. Then \Eq{SecondDiff1} becomes
\beq{SecondDiff2}
	\Cs (\xi_{t+1}-2\xi_t+ \xi_{t-1}) - \Ca(\xi_{t+1}-\xi_{t-1})  = \nabla U(\xi_t).
\eeq
In this form the map closely resembles a pair of second order differential equations.

Indeed in a neighborhood of the origin in the phase space, $(\xi,\eta)$, and in the space of the new parameters, $(a,b,c, \nu)$, \Eq{SecondDiff2} approaches a Lagrangian system of ODEs.
To see this, formally introduce a parameter $h$, and  scale
$$
	(a,b,c,\nu) \to (h^4a,h^4 b,h^2c, h \nu).
$$
Here $h \nu$ represents the deviation from symmetry, so that $C \to \Cs + h \Ca$.
Then in the limit $h \to 0$, the second difference equation \Eq{SecondDiff1} limits on a system of ODEs in a scaled
time  $\tau = ht$, and a new variable
$$
	\xi_t \to h^2 q(\tau) ,
$$
This scaling implies that for the potential \Eq{UPotential}, $\nabla_\xi U(\xi) \to h^4 \nabla_q U(q)$.
Moreover, as $h \to 0$, the second difference $\xi_{t+1}-2\xi_t+ \xi_{t-1} \to h^4 \ddot{q}(\tau) + \cO(h^5)$ and the first difference $\xi_{t+1}-\xi_{t-1} \to 2h^3\dot{q}(\tau) +\cO(h^4)$. Substituting these into \Eq{SecondDiff2}, gives the limiting system
\beq{ODEs}
	\Cs \ddot{q}  - 2 \Ca \dot{q} = \nabla U(q) ,
\eeq
as $h \to 0$.
Thus the symmetric part of $C$ corresponds to a mass matrix, multiplying the acceleration. By contrast, $\Ca$ corresponds to a Coriolis-like force, which is proportional to the velocity.

The system \Eq{ODEs} is obtained from the Lagrangian
\beq{Lagrangian2}
	L(q,\dot{q}) = \tfrac12 \dot{q}^T \Cs \dot{q} + q^T \Ca \dot{q} + U(q).
\eeq
To convert \Eq{ODEs} into a Hamiltonian system define the canonical momenta
$$
	p(\tau) = \frac{\partial L}{\partial \dot{q}} = \Cs\dot{q} - \Ca q  ,
$$
giving a Hamiltonian, $H = p \dot{q} - L$, that has Coriolis and centripetal terms:
\beq{Hamiltonian}
	H(q,p) = \tfrac12 p^T C^{-s} p + p^T C^{-s} C^{a} q - \tfrac12 q^T (\Ca C^{-s} \Ca) q - U(q) .
\eeq
Thus $\Cs$ is the mass matrix, and $U$ is the \textit{negative} of the potential energy. The antisymmetric matrix $\Cs$ contributes both a Coriolis-like term, bilinear in $q$ and $p$, and a centripetal-like term, quadratic in $q$.
We will use this interpretation, for the case of symmetric $C$, in \Sec{Symmetric}.

%%%%%%%%%%%%%%%%%%%%%%%%%%%%%%%%%%%%%%%%%%%%%%%%%%%%%%%%%%%%%%%%%%%%%%%%%%%%%
\section{Quadfurcation}\label{sec:Quadfurcation}
%%%%%%%%%%%%%%%%%%%%%%%%%%%%%%%%%%%%%%%%%%%%%%%%%%%%%%%%%%%%%%%%%%%%%%%%%%%%%
From the general theory \cite{Moser94} we know that \Eq{MoserMap} and hence \Eq{ShiftedMap} has at most four (isolated) fixed points. On the codimension-three surface $a = b = c = 0$ in parameter space there is a single fixed point (unless $\eps_2 = 0$). As we will see below, there are sectors in parameter space near this surface for which there are no fixed points, and sectors for which there are four. It seems appropriate to call the creation of four fixed points from none a \textit{quadfurcation}. In some cases a quadfurcation can be analogous to a simultaneous pair of co-located saddle-center bifurcations; however, the stabilities of the resulting fixed points are usually not those of a pair of decoupled area-preserving maps, namely the Cartesian product of \twoD{} saddles and centers.

In the following subsections we study the fixed points, their stability, and the structure of the region of phase space around the elliptic fixed points that contains bounded orbits.

%%%%%%%%%%%%%%%%%
%% Fixed Points
%%%%%%%%%%%%%%%%%
\subsection{Fixed Points}\label{sec:FixedPoints}

The coordinates, $\xi^*$, of the fixed points of \Eq{ShiftedMap} are critical points of the cubic polynomial \Eq{UPotential}. Several contour plots of $U(\xi)$ are shown in \Fig{PotentialContour}. Critical points satisfy the equations
\beq{ShiftedFP}
	0 = \nabla U(\xi^*) = \begin{pmatrix}
						a +  c \xi_1^* + 3 \eps_2 \xi_1^{*2} + \xi_2^{*2} \\
						b + 2\xi^*_1 \xi^*_2 \end{pmatrix}.
\eeq
Note that the positions are independent of the matrix $C$, though the momenta, determined by $\eta^* = C\xi^*$,
depend on the full matrix. Note that if $\nabla U(\xi^*) = 0$ for parameters $(a,b,c)$, then it is also zero
at the point $-\xi^*$ for parameters $(a,b,-c)$ and at the point $(\xi_1^*, -\xi_2^*)$ for $(a,-b,c)$. Thus we can restrict attention to $b,c \ge 0$.

The case $a=b=c=0$ is an organizing center for the solutions of
\Eq{ShiftedFP}. In this case the second component immediately implies that either $\xi_1^* = 0$
or $\xi_2^* = 0$. Then, whenever $\eps_2 \neq 0$, the first implies that both $\xi_1^* = \xi_2^*
= 0$. We call this the \textit{quadfurcation point}. Since the matrix elements $(\alpha,\delta,
\mu)$ are still free parameters, it occurs on a codimension-three surface in the six-dimensional
parameter space. The off-diagonal elements of the matrix, $\beta$ and $\gamma$, are then fixed up to exchange
by the condition det$(C)=\eps_1$, \Eq{betagamma}. In the parameterization
\Eq{ShiftedMap} the quadfurcation surface is just the three-plane $a=b=c=0$. In Moser's original
parameterization, this surface is determined by \Eq{AStar}.

%%%%%
\InsertFig[b]{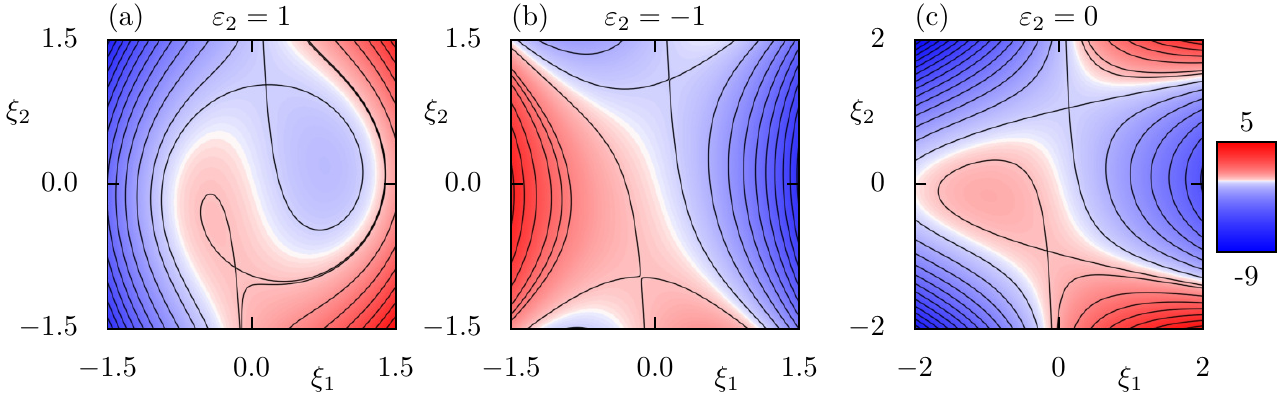}{Contour plots of the potential $U$ \Eq{UPotential} for $(a,b,c) = (-1,-0.3,-1)$. (a) $\eps_2 = 1$. Here there are four critical points, implying four fixed points of \Eq{ShiftedMap}. (b) $\eps_2=-1$, two critical points. (c) $\eps_2 = 0$, three critical points.}{PotentialContour}{Note:WidthIsDeterminedFromPdfFile}
%%%%%

%%%%%
\InsertFig[b]{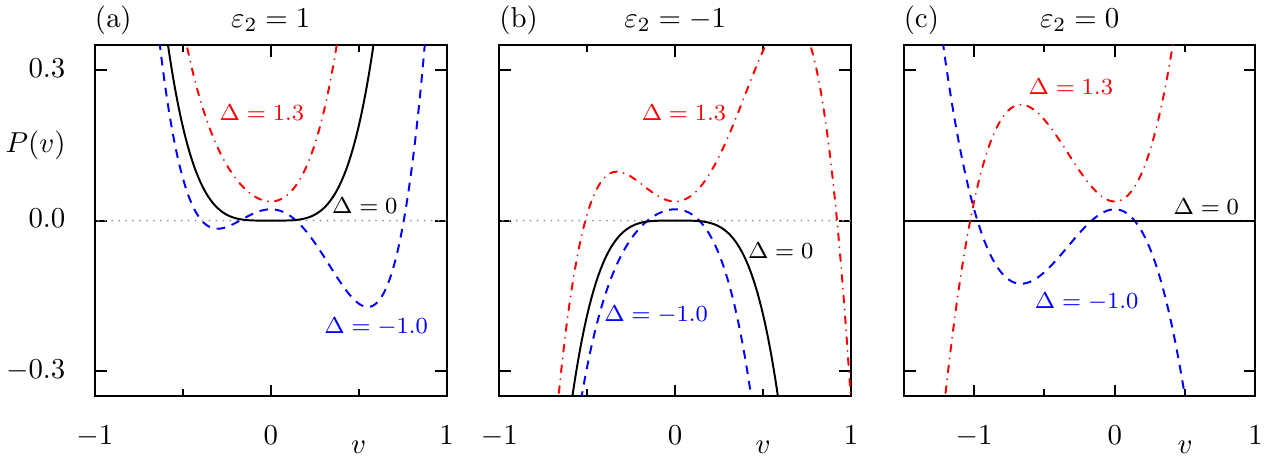}
{The polynomial $P(v)$ \Eq{PofV} along the curve $(a,b,c) = \Delta(1,0.3,1)$ as $\Delta$ varies.
(a) $\eps_2 = 1$ where four roots are created when $\Delta$ decreases through zero;  (b) $\eps_2 = -1$ where there are two roots for any $\Delta \neq 0$; and (c) $\eps_2 = 0$ with one root for $\Delta >0$, infinitely many at $\Delta = 0$, and three for $\Delta < 0$.}{PofV}{Note:WidthIsDeterminedFromPdfFile}
%%%%%

More generally if $b \neq 0$ then \Eq{ShiftedFP} implies that $\xi_1^* \neq 0$, so
\beq{xi2Ofxi1}
	\xi_2^* = -\frac{b}{2\xi_1^*}  .
\eeq
Substituting into the first component of \Eq{ShiftedFP} then shows that $\xi_1^*$ must be a root of
the scalar polynomial
\beq{PofV}
	P(v;a,b,c,\eps_2) = 3\eps_2v^4 + c v^3 + a v^2 + \tfrac14 b^2 .
\eeq
When $\eps_2 \neq 0$ this polynomial is quartic, and so has at most four roots. Since $P$ has no linear term, it has exactly one root in $\bC$ only when $a = b = c = 0$, on the quadfurcation set. When $\eps_2 = 0$, $P(v)$ is at most cubic, and there are at most three isolated roots. Several examples are shown in \Fig{PofV}.

There are various regions in the parameter space $(a, b, c)$ that have different numbers of fixed points.
We now determine the bifurcation sets, which separate these regions:
to find these sets when $\eps_2 \neq 0$, it is easiest to solve for the surfaces on which there are double roots, i.e., $P(v)=0$ and $P'(v)=0$.
First, $P$ always has a critical point, $P'(v)=0$,
at $v=0$, and it has two more critical points if
\beq{aPlus}
	\eps_2a < a_+(c) \equiv \tfrac{3}{32} c^2 .
\eeq
Eliminating $v$ from the two equations $P(v) = P'(v) =0$ gives the discriminant
$$
	b^2 \left[ 1728\eps_2b^4+ 9( -3c^4+48\eps_2 ac^2-128a^2 ) b^2-16a^3c^2+192\eps_2a^4 \right]
			= 0.
$$
Thus there are double roots at $b = 0$, and on the surfaces $b = \pm \sqrt{D_\pm(a,c)}$, where
\beq{DPlusMinus}
	D_\pm(a,c) = \frac{\eps_2}{384}\left( 3(c^2-8\eps_2 a)^2 - 64a^2 \pm
	     \frac{|c|}{\sqrt{3}} (3c^2-32\eps_2a)^{3/2}\right) .
\eeq
For these surfaces to be real, the radical in \Eq{DPlusMinus} must be real,
i.e., \Eq{aPlus} must be satisfied.
Moreover letting
\beq{aMinus}
	a_-(c) \equiv \tfrac{1}{12}c^2,
\eeq
then $D_+(\eps_2 a_- ,c) = 0$ and $D_-(0,c) = 0$.
To define real-valued functions let
\bsplit{bPlusMinus}
	b_+(a,c) &= \left\{ \begin{array}{cl}
						\sqrt{D_+(a,c)}, &  a < \eps_2 a_-, \mbox{ and } \eps_2 a < a_+ \\
						0,		         & \mbox{ otherwise }
				  \end{array} \right. ,\\
	b_-(a,c) &= \left\{ \begin{array}{cl}
						\sqrt{D_-(a,c)}, &  a < 0, \mbox{ and } \eps_2 a < a_+ \\
						0,		         & \mbox{ otherwise }
				  \end{array} \right. .
\esplit
The resulting surfaces are shown in \Fig{BifSurfaces}.

%%%%%
\InsertFigTwo{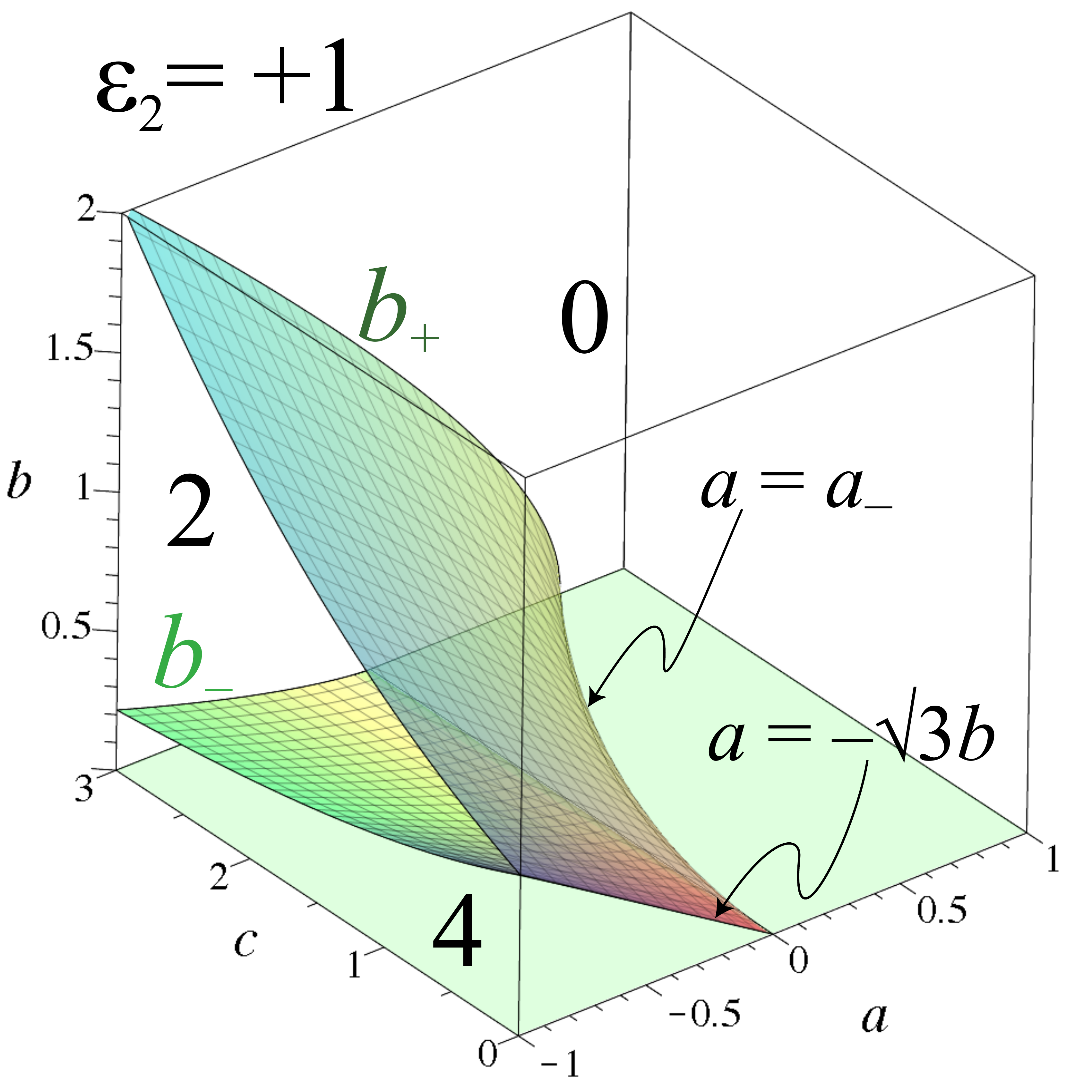}{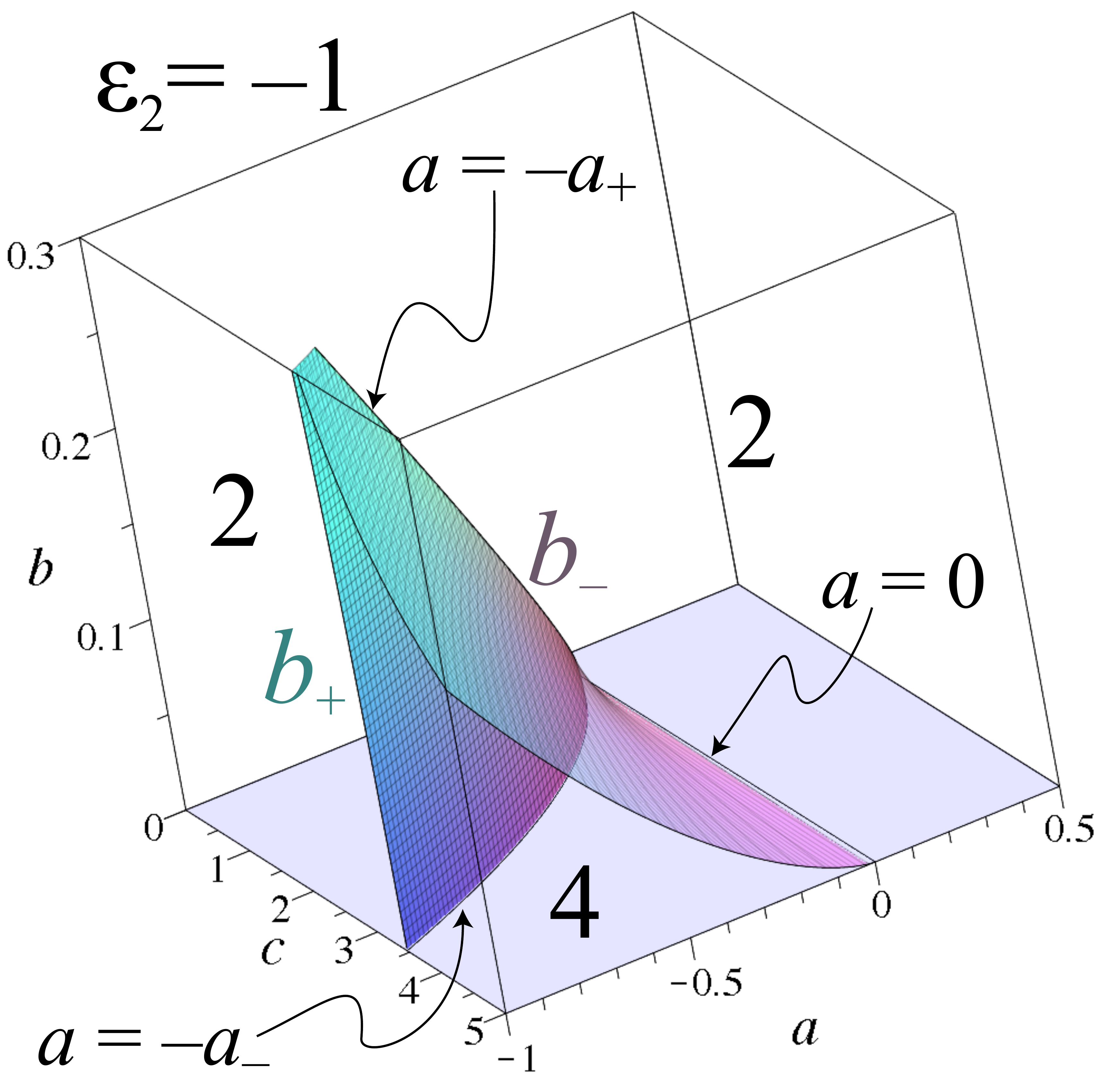}{Surfaces in $(a,b,c)$ for which there are double roots of $P$ \Eq{PofV} (a) when $\eps_2 = 1$ and (b) when $\eps_2 = -1$. For $\eps_2 =1$, when $b$ is sufficiently positive there are no real solutions, as $b$ decreases through the $b_+$ surface \Eq{bPlusMinus}, two solutions are  created, and finally when it passes through through the $b_-$ surface there are four solutions. For $\eps_2 = -1$, there are two solutions
if $b$ is sufficiently large. An additional two solutions are created upon moving through either of the $b_{\pm}$ surfaces.}{BifSurfaces}{7.8cm}
%%%%%

If $\eps_2 = 1$, then when $b > b_+$ there are no real roots. At the upper surface, $b = b_{+}(a,c)$, which is nonzero for $a < a_-$, two roots are created. Two additional roots are created upon crossing $b_-(a,c)$, which is nonzero for $a < 0$, see \Tbl{NumFixed}. The two surfaces $b_\pm(a,c)$ intersect at $c=0$ on the line
$a = -\sqrt{3} |b|$. Crossing this codimension-two set $b_- = b_+$ at $c = 0$ and moving into the region $ b < b_-$ thus corresponds to the simultaneous creation of four fixed points at two different locations, i.e., to a pair of simultaneous saddle-center bifurcations, as we will see in \Sec{Stability}.

If $\eps_2 = -1$, then \Eq{PofV} has four real roots only if $-a_+< a < 0$ and $b_+ < b < b_-$. Note that $b_+$ is nonzero only when $-a_- < a < 0$.
Inside cusp-like shape formed from the $b_\pm$ surfaces,
as shown in \Fig{BifSurfaces}(b), there are four roots.
Going outwards from this region, by either crossing
$b_-$ or $b_+$, two solutions disappear in a saddle-center
bifurcation.
Thus on the surfaces $b = b_+$ or $b = b_-$ there are three roots
(one of them with multiplicity 2).
When these surfaces merge,  on the curve $a=-a_+$ and $b = b_\pm(-a_+,c) = c^2/32$, there are two fixed points:
$$
	(\xi_1^*, \xi_2^*) = \left\{\begin{array}{ll}
						\frac{c}{8}( 1,-1),  &\mbox{ (multiplicity 3)} \\
	                     \frac{c}{24} (-1, 9), & \mbox{ (multiplicity 1)}
	                     \end{array} \right.  \quad (\eps_2 = -1, a= -3b = -a_+).
$$

The cases $b=0$ and $c=0$ require special treatment.
When $c=0$ but $b\neq 0$, the polynomial \Eq{PofV} has no cubic term and
the fixed points can be solved for explicitly:
\beq{cZeroFPs}
	\xi_1^* = \pm \sqrt{-\tfrac16 \eps_2 \left(a \pm \sqrt{a^2-3\eps_2b^2}\right)}
	\quad\quad (c = 0) ,
\eeq
where $\xi_2^*$ is then obtained from \Eq{xi2Ofxi1}.
Note that there are four possible points
here, with choices for the outer $\pm$ and the inner $\pm$.
This equation gives real solutions only when both square roots are real.
When $\eps_2 = 1$, \Eq{cZeroFPs} gives four real solutions if $a < -\sqrt{3}|b|$.
On the boundary $a=-\sqrt{3}|b|$, these four solutions are created in two pairs, at
$$
	(\xi_1^*,\xi_2^*) = \pm \sqrt{-\frac{a}{6}}\left(1,-\sqrt{3}\sgn(b))\right)
			\quad (\eps_2 = 1, a=-\sqrt{3}|b|,c=0) .
$$
These pairs merge at the quadfurcation point $a=b=c=0$. If $\eps_2 =-1$ then only the inner $+$ sign choice is valid and \Eq{cZeroFPs} gives two real solutions whenever $a$ or $b \neq 0$. Table~\ref{tbl:NumFixed} delineates the possibilities.

If $b = 0$, then \Eq{ShiftedFP} implies that either $\xi_1^* = 0$ or $\xi_2^* = 0$. The first component of \Eq{ShiftedFP} is then trivially a quadratic. In this case the four solutions are
\beq{bZeroFP}
   (\xi_1^*,\xi_2^*) = \left\{ \begin{array}{l}
   \left(0,  \pm \sqrt{-a} \right )\\
	\left(\frac{\eps_2}{6}( - c \pm \sqrt{c^2 - 12\eps_2 a }), 0 \right )
	\end{array}	\right.  \quad  (b = 0).
\eeq
Note that when $\eps_2 = 1$ there are four real fixed points whenever $a <0$, and two
in the range $0 < a < a_-$. If $\eps_2 = -1$, then the first pair is real when $a \le 0$, and the second pair is real only if $ a \ge -a_-$. Thus there are four real fixed points when $-a_- < a < 0$.

\newcommand\T{\rule{0pt}{2.6ex}}       % Top strut
\newcommand\B{\rule[-1.2ex]{0pt}{0pt}} % Bottom strut

%%%%%%%%%%%%%%%
\begin{table}[tbp]
   \centering
   \begin{tabular}{@{} r| cc||c|c|c|c|c @{}}
    \multicolumn{3}{c||}{} &  \multicolumn{5}{c}{Number of Fixed Points} \\
      $\eps_2$ & $b$ &$c$    & $0$ & $1$  & $2$ & $3$ & $4$ \\ \hline \hline\hline
      1     &   &   & $b>b_+$ & $b=b_+$ & $b \in (b_-,b_+)$       & $b=b_-$& $b < b_-$       \T\B \\ \cline{2-8}
            &$0$&   & $a > a_-$         & $a = a_-$ & $0< a < a_-$      & $a=0 $ & $a<0$            \T\B\\ \cline{2-8}
            &   &$0$& $a>-\sqrt{3}|b|$  &           & $a = -\sqrt{3}|b|$&        & $a<-\sqrt{3}|b|$ \T\B\\ \cline{2-8}
            &$0$&$0$& $a > 0$           & $a = 0$   &                   &        & $ a <0$          \T\B\\
     \hline\hline\hline
      $-1$    &   &  &       & &  $a \notin (-a_+,0]$,    & $b=b_+$,   & $a \in (-a_+,0)$,      \T\B\\
      		  &   &  &       & &  or $b\notin [b_+,b_-]$ & or $b=b_-$  &  and $ b\in (b_+,b_-)$  \T\B\\ \cline{2-8}
            &$0$&  &         & & $a\notin [-a_-,0]$  &  $a=-a_-,0$& $a\in(-a_-,0)$         \T\B\\ \cline{2-8}
          &   &$0$&          & & always               &                   &  \T\B\\ \cline{2-8}
           &$0$&$0$&         & $a=0$                      & $a \neq 0$        & &  \T\B\\
    \hline\hline\hline
      0     &   &  &         & $a > -a_0$ &  $a = -a_0$ &$a < -a_0$ &  \T\B\\ \cline{2-8}
            &$0$&  &         & $a=0$      & $a\neq 0$   &           &  \T\B\\ \cline{2-8}
            &   &$0$&$a\ge 0$&            & $a < 0$     &           &  \T\B\\ \cline{2-8}
            &$0$&$0$&        & $a\neq0$   &             &           &  \T\B
 \end{tabular}
   \caption{Number of fixed points of the map \Eq{ShiftedMap} depending upon $\eps_2$ and parameters
   $(a,b,c)$. Since this number is an even function of $b$ and $c$, we can assume that both are nonnegative.
   The functions $b_\pm(a,c)$ are given by \Eq{bPlusMinus}, and $a_\pm(c)$ by \Eq{aPlus} and \Eq{aMinus}.
   The additional rows delineate special cases when either $b$ or $c$ or both are zero.
   \label{tbl:NumFixed}}
\end{table}
%%%%%%%%%%%%%%%

For $\eps_2 =0$ and $c \neq 0$, the polynomial \Eq{PofV} is cubic, so there is always at least one fixed point.
The critical points of $P$ are at $v = 0$ and $v = -2a/3c$, and critical values $P(0) = \tfrac14 b^2 \ge 0$ and
$P(-2a/3c) = \tfrac{4}{27}\frac{a^3}{c^2}+ \tfrac14 b^2$. Thus there are three fixed points when
\beq{azero}
	a < -a_0 \equiv -3\left|\frac{bc}{4}\right|^{2/3} .
\eeq
If $c = 0$, then the polynomial \Eq{PofV} is at most quadratic. It has two solutions if $a < 0$.
Special cases are again shown in \Tbl{NumFixed}.
Finally, when $a=b=c=0$, there is a line of fixed points at $\xi_2 = 0$.
This case (not shown in \Tbl{NumFixed}) is the only one for which there are infinitely many fixed points.

%%%%%%%%%%%%%%%%%
%% Stability
%%%%%%%%%%%%%%%%%
\subsection{Stability}\label{sec:Stability}

The stability properties of fixed points of the map \Eq{ShiftedMap} are most easily computed using the second difference form \Eq{SecondDiff1}. Linearization about a fixed point gives the $2 \times 2$ eigenvalue problem
\beq{EVProblem}
	\left(\lambda C^T + \lambda^{-1} C \right) q = W q ,
\eeq
where $W = C +C^T + D^2U(\xi) $ is the symmetric matrix
\bsplit{WMatrix}
	W 	&= \begin{pmatrix} 2\alpha  & \mu  \\
							 \mu  & 2\delta  \end{pmatrix} +
							 \begin{pmatrix} c + 6\eps_2\xi_1 &  2\xi_2 \\
							 2\xi_2 &  2\xi_1  \end{pmatrix} .
\esplit
Given the coordinate eigenvector $q$, the momentum components are $p = \lambda^{-1}Cq$. A similar analysis can be used, more generally, for a period-$n$ orbit, see \cite[eqs. (21)-(22)]{Kook89}.

Thus for a nontrivial solution of \Eq{EVProblem}, the $2\times 2$ matrix
$$
	N(\lambda) = \lambda C^T+\lambda^{-1} C-W
$$
must be singular.
Since the map is symplectic its eigenvalues must satisfy the reflexive property:  if $\lambda$ is an eigenvalue, then so is $\lambda^{-1}$. This follows for \Eq{EVProblem} because $W^T = W$ implies that $N^T(\lambda) = N(\lambda^{-1})$.
As a consequence the characteristic polynomial can be written as a quadratic in the partial trace $\rho = \lambda + \lambda^{-1}$:
\begin{equation}\label{eq:ReducedPolyn}
	\det(N(\lambda)) = \eps_1 (\rho^2 - A \rho + B-2)
\end{equation}
where we recall that $\eps_1 = \det{C}$. The parameters $A,B$ are
Broucke's stability parameters \cite{Broucke69, Howard87}.
More generally, these parameters are determined by the linearized map $DM$ at a fixed point, by
$A = \tr{DM}$ and
$B = \tfrac12 \left[(\tr{DM})^2 - \tr{DM^2}\right]$;
equivalently, in terms of the eigenvalues $\rho_{1,2}$ of the
reduced characteristic polynomial one has
$A = \rho_1 + \rho_2$ and $B=\rho_1 \rho_2 + 2$, or explicitly
\begin{equation}\label{eq:PartialTraces}
 \rho_{1, 2} = \tfrac{1}{2} \left(A \pm \sqrt{A^2 + 8 - 4B}\right)
            \equiv \lambda_{1,2} + \lambda_{1,2}^{-1},
\end{equation}
where $\lambda_{1,2}, \lambda_{1,2}^{-1}$  are the two reciprocal pairs of
eigenvalues of the characteristic polynomial of the linearized map.

The $(A,B)$-plane is divided into seven stability regions as shown in \Fig{ABPlane}.
These are bounded by the saddle-center ($SC$), and period-doubling ($PD$) lines:
\bsplit{BifCuvres}
	SC &= B-2A+2 = 0, \\
	PD &= B+2A+2 = 0,
\esplit
on which there is a pair of eigenvalues at $+1$ or $-1$, respectively,
and the Krein parabola ($KP$)
\beq{KreinBifs}
	KP = B-A^2/4-2 = 0
\eeq
on which there are double eigenvalues on the unit circle (for $B<6$ and $|A| < 4$) or real axis (for $B \ge 6$ and $|A|\ge 4$.). The point $(A,B) = (4,6)$ corresponds to four unit eigenvalues.
The seven stability regions with different types of linearized dynamics
around the fixed point are labeled by
combinations of E (elliptic), H (hyperbolic),
and I (inverse hyperbolic), each involving a pair
of eigenvalues $(\lambda, 1/\lambda)$, and
in the CU (complex unstable) region, where $KP>0$,
there is a complex quartet of eigenvalues.

%%%%%
\InsertFig[t]{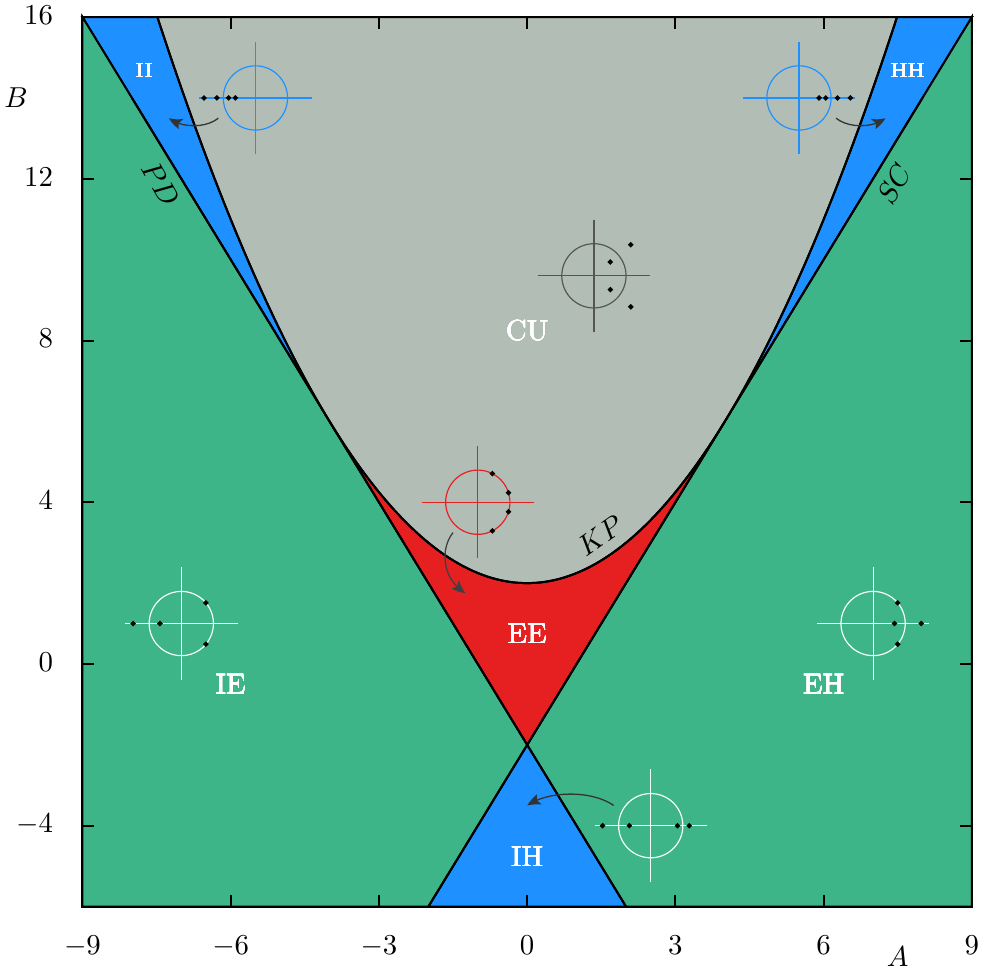}{Stability regions for a
reflexive, quartic, characteristic polynomial.
There are seven regions, EE, EH, IE, IH, II, HH, and CU,
with the configuration of the four eigenvalues
as shown in the representative complex plane insets.}
{ABPlane}{Note:WidthIsDeterminedFromPdfFile}
%%%%%

For a general symmetric $W$, the stability parameters are
\begin{align*}
	A &= \eps_1(w_{22} \alpha + w_{11}\delta - w_{12} \mu) ,\\
	B &= \eps_1(\det(W) -(\beta-\gamma)^2) + 2 .
\end{align*}
For the matrix \Eq{WMatrix} these become
\bsplit{AB}
	A &= 4 + \eps_1\left[ \delta  c -(\beta-\gamma)^2+(6\delta\eps_2+2\alpha)\xi_1 -2\mu\xi_2\right], \\
	B &= 6 + \eps_1 \left[2 (\delta c - (\beta-\gamma)^2) + 12\eps_2(\xi_1+\delta)\xi_1 +
			2(2\alpha +  c)\xi_1 - 4(\xi_2+\mu)\xi_2)\right] ,
\esplit
and the saddle-center and period-doubling parameters are
\bsplit{SCPD}
	SC &= \eps_1 \det(D^2U) = \eps_1(12\eps_2\xi_1^2 + 2c \xi_1 - 4 \xi_2^2),\\
	PD &= 16+ 4\eps_1\left(\delta c -(\beta-\gamma)^2 + 3\eps_2\xi_1^2 + (6\delta\eps_2 + 2\alpha + \tfrac12 c)\xi_1 -\xi_2^2-2\mu \xi_2\right).
\esplit
In particular, note that the sign of $SC$ depends on the fixed points only through the sign of the Hessian of the potential \Eq{UPotential}; for example, if $\eps_1 = 1$, then $SC >0$ at extrema and $SC < 0$ at saddle points of $U$.
Finally, the Krein parameter is
\bsplit{Krein}
	KP =& -\tfrac14(\delta c - (\beta-\gamma)^2)^2  \\
	   &+  (c\delta - (\beta-\gamma)^2) (\mu\xi_2 -(3\delta\eps_2 +\alpha)\xi_1) + 2\eps_1c\xi_1 \\
	   &+(12\eps_1\eps_2 -(3\delta\eps_2+\alpha)^2)\xi_1^2
	    +2\mu(3\delta\eps_2+\alpha)\xi_1\xi_2
	    -(\mu^2+4\eps_1)\xi_2^2.
\esplit

At the quadfurcation, where $a=b=c=\xi_1 = \xi_2 = 0$, \Eq{AB} gives
\bsplit{ABQuad}
	A^\txtQ &= 4-\eps_1(\beta-\gamma)^2 ,\\
	B^\txtQ &= 6 -2\eps_1(\beta-\gamma)^2 .
\esplit
This implies that the quadfurcation point lies on the saddle-center line; indeed from \Eq{SCPD}, $SC = 0$ at this
point. When $\eps_1 =1$ ($\eps_1 = -1$), then $(A^\txtQ,B^\txtQ)$ lies below and to the left of (above and
to the right of) the point $(4,6)$. The quadfurcation occurs at $(4,6)$ only when $\beta=\gamma$, i.e., when the
matrix $C$ is symmetric,  see the discussion in \Sec{Symmetric} below. The quadfurcation occurs below the
period-doubling line, i.e., for $PD < 0$, only if $\eps_1 = 1$ and $|\beta-\gamma| > 2$.

More generally, since fixed points are critical points of $U$, \Eq{ShiftedFP}, they can be created or destroyed
only when $\det(D^2U) = 0$, which is equivalent to $SC = 0$ by \Eq{SCPD}. This can also be seen upon computing the \textit{resultant} of $SC$ and $P$ \Eq{PofV}---recall that the resultant gives the set of parameters on which two polynomials simultaneously vanish. This resultant is proportional to $b^4 (b^2-b_-^2)(b^2-b_+^2)$. Of course, this is what we saw in \Fig{BifSurfaces}---pairs of fixed points are created or destroyed upon crossing the surfaces $b = b_\pm$  \Eq{bPlusMinus}.

%%%%%%%
%%%% Quadfurcation along lines
%%%%%%%
\subsection{Quadfurcation along a Line in Parameter Space}\label{sec:QuadLines}

Near the quadfurcation, if we assume that $a,b,c = \cO(\Delta)$ for $\Delta \ll 1$, then $\xi_i= \cO(\sqrt{\Delta})$,  the cubic term involving $c$ in \Eq{PofV} is negligible to lowest order, and the fixed points are given by \Eq{cZeroFPs} to $\cO(\sqrt{\Delta})$.
Substitution into the stability criteria then gives
\begin{subequations} \label{eq:OrderDeltaStab}
\begin{align}
	SC &= \mp 4\eps_1 \sqrt{a^2-3\eps_2 b^2} + \cO(\Delta^{3/2})
             \label{eq:OrderDeltaStab-SC}, \\
	PD &= 16 - 4\eps_1(\beta-\gamma)^2 + 8 \eps_1[( 3\delta \eps_2 + \alpha)\xi_1 -\mu \xi_2]\sqrt{\Delta}+\cO(\Delta) ,
              \label{eq:OrderDeltaStab-PD} \\
    KP&=
       (\beta-\gamma)^2\left[ -\tfrac14 (\beta-\gamma)^2+ (3\delta\eps_2 +\alpha)\xi_1-\mu \xi_2\right] + \cO(\Delta).
    		  \label{eq:OrderDeltaStab-K}
\end{align}
\end{subequations}
Note that the $\mp$ signs in \Eq{OrderDeltaStab-SC}
correspond to the inner $\pm$ in \Eq{cZeroFPs}, the sign inside the square root.
Using these results, we can get an overview of all possible
stability scenarios of the fixed points created in a quadfurcation,
see \Tbl{overview} and \Fig{ABDiagrams}. As the quadfurcation point shifts along the $SC$ line, different
stabilities occur, but since $SC \sim \Delta$ and generically $PD-PD^\txtQ \sim \sqrt{\Delta}$,
the branches emerge tangentially to the $SC$ line
(the $\cO(\sqrt{\Delta})$ term could vanish, but this is exceptional).
Moreover, the sign of $A-A^\txtQ$ depends on
the choice of the outer $\pm$ sign in \Eq{cZeroFPs}, so this pair of fixed points form a parabolic
curve that is tangent to the $SC$ line at the quadfurcation point.

%%%%%%%%%%%%%%%
\begin{table}[b]
\begin{center}
\begin{tabular}{c||c|c||c|c}
$(A^\txtQ, B^\txtQ)$  & \multicolumn{2}{c||}{Condition}    &  \multicolumn{2}{c}{Fixed Points and Stability} \\
     on $SC$    &   $\eps_1$ & $|\beta-\gamma|$   &  $\eps_2=1$, $a < -\sqrt{3}|b|$  & $\eps_2=-1$, $a$ or $b\neq0$  \\\hline
$ > (4, 6)$        & $-1 $  & $\neq 0$            &  2 EH $+$ 2 HH & 2 HH  \\
$ = (4, 6)$        & $\pm 1$ & $0$                &    see \Sec{Symmetric}   & see \Sec{Symmetric}   \\
$ < (4, 6)$        & $1$    &  $<2$           &  2 EE $+$ 2 EH & 2 EH  \\
$ = (0,-2)$		   & $1$    & $2$ & IE $+$ EE $+$ IH $+$ EH & IH $+$ EH \\
$ < (0, -2)$       & $1$   &$> 2$ &  2 IE $+$ 2 IH & 2 IH   \\
\end{tabular}
\end{center}
\caption{Overview of the location of the quadfurcation along the $SC$ line depending on the the value of $\eps_1$ and
the asymmetry of $C$. Stabilities of the fixed points are shown in the last two columns for a path of the form $(a,b,c) = \Delta(a^*,b^*,c^*)$, which has a quadfurcation at $\Delta = 0$.
For $\eps_2 = 1$, four fixed points are created as $a$ becomes negative if $a < -\sqrt{3}|b|$.
Their stabilities are shown in column four.
When $\eps_2=-1$, two fixed points exist whenever $a$ or $b \neq 0$ and collide at $\Delta = 0$;
their stabilities are shown in the last column.
\break
}
\label{tbl:overview}
\end{table}
%%%%%%%%%%%%%%%

Finally, note that when $\beta \neq \gamma$, then $KP < 0$ at the quadfurcation, and hence a direct
transition to the complex unstable (CU) region is only possible in the symmetric case for
which $(A^\txtQ,B^\txtQ) = (4,6)$; this is discussed in \Sec{Symmetric}.

%%%%%
\InsertFig{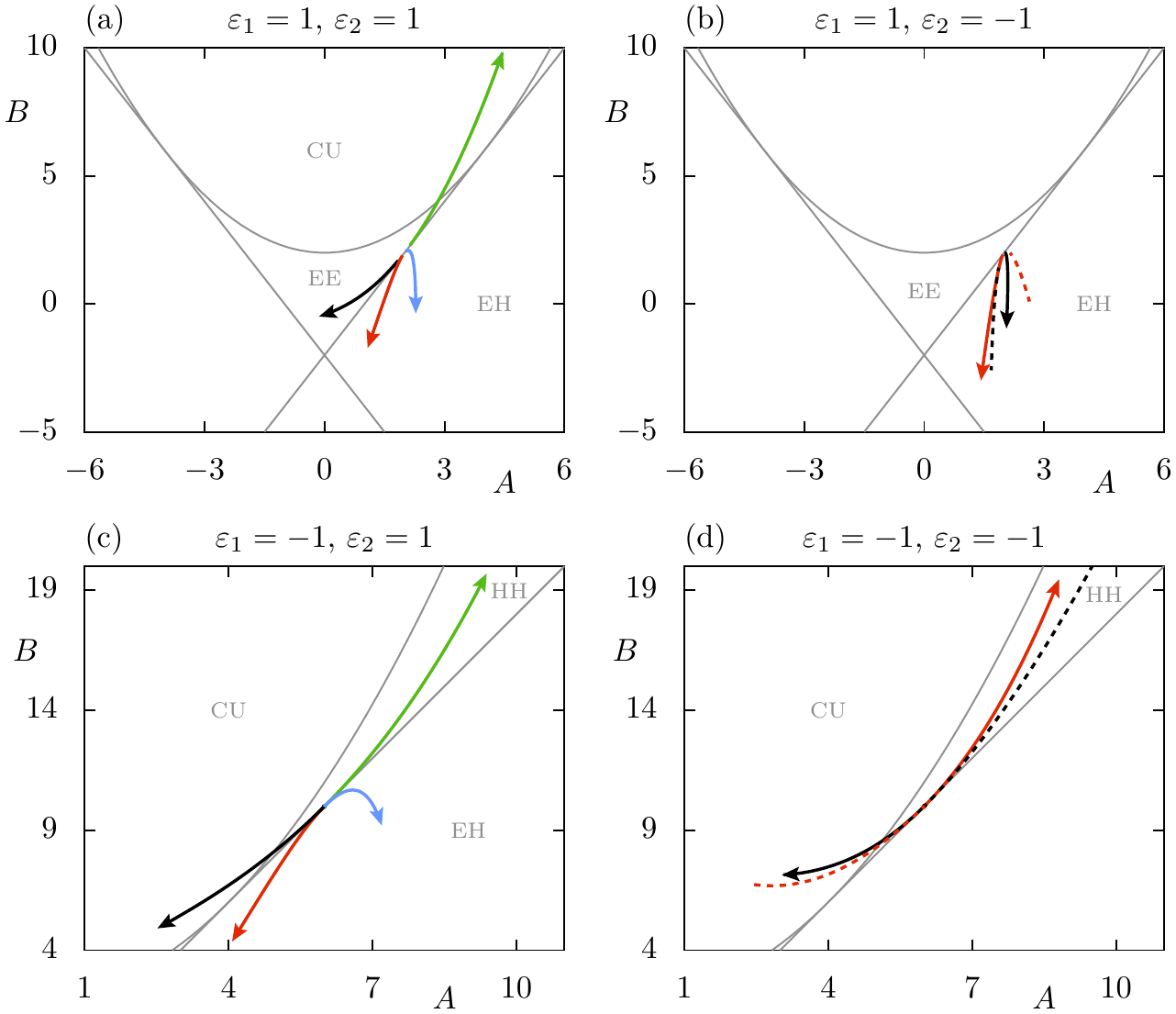}{Stability of fixed points for
$(a,b,c) = \Delta(1.5,0.5,1)$ with $\Delta \in [-1, 1]$
and different choices of $C$,  $\eps_2$, and $\eps_1$:
(a, b) Case $(\alpha, \mu,\delta, \eps_1) = (1,0.1,0.5,1)$.
Since $|\beta-\gamma| \approx 1.42$, the quadrupling occurs at
$(A^\txtQ,B^\txtQ) = (1.99,1.98)$.
(a) $\eps_2=+1$: creation of four fixed points in the transition \Eq{TwoEETwoEH} for $\Delta <0$.
(b) $\eps_2=-1$: the transition \Eq{TwoEHtoTwoEH}, before and after the bifurcation one has two fixed points.
(c, d) Case  $(\alpha, \mu, \delta) = (1, 2.0, -0.5)$, $\eps_1=-1$,
giving  $(\beta, \gamma) \simeq (1.71, 0.29)$
and the quadrupling occurs at $(A^\txtQ,B^\txtQ) = (6, 10)$.
(c) $\eps_2=1$: creation of four fixed points \Eq{TwoEHTwoHH} for $\Delta <0$.
(d) $\eps_2=-1$: the transition \Eq{TwoHHtoTwoHH} two fixed points before and after the bifurcation.
The arrows indicate the direction towards more negative $\Delta$
and in (b, d) the branches for $\Delta>0$ are shown as dashed lines.}
{ABDiagrams}{Note:WidthIsDeterminedFromPdfFile}
%%%%%

For $\eps_2=1$ the basic structure
is the one shown in the fourth column of \Tbl{overview} and \Fig{ABDiagrams}(a, c):
four new branches emerge from a point on the saddle-center line.
In each case, two of the fixed points are above
($SC >0$), and two are below ($SC <0$), this line.
The implication is that
when $\eps_1 = 1$, and $|\beta-\gamma| > 2$ the quadfurcation occurs below the $PD$ line and
corresponds to the transition
$$
	\emptyset \to 2\mbox{ IE} + 2\mbox{ IH} ,
$$
i.e., two of the created fixed points are of type IE and two of type IH (this is not shown in \Fig{ABDiagrams},
but compare with \Fig{ABPlane}).
Perhaps the most interesting quadfurcation creates stable fixed points. This occurs for $\eps_1 = 1$ and $0 < |\beta-\gamma| < 2$, where we have the transition
\beq{TwoEETwoEH}
	\emptyset \to 2\mbox{ EE} + 2\mbox{ EH} .
\eeq
This is the case shown in \Fig{ABDiagrams}(a): as $\Delta$ decreases from zero, the two created EE points move
along the green and black curves in the figure and the two EH points move along the red and blue curves.

Note that one of the EE points in this figure eventually undergoes a Krein bifurcation, moving into the CU region.
The implication is that the Krein signature of this point must have been indefinite when it was created in the quadfurcation, since this signature is constant under parameter variations so long as the stability remains in the interior of the EE-region in \Fig{ABPlane} \cite[Sec.~III]{Howard87}.

Finally, when $\eps_1=-1$ the quadfurcation point is above $(A,B) = (4,6)$ whenever $\beta \neq \gamma$, so the transition is
\beq{TwoEHTwoHH}
	\emptyset \to 2\mbox{ EH} + 2\mbox{ HH} .
\eeq
This case is shown in \Fig{ABDiagrams}(c). Again, a Krein bifurcation, HH $\to$ CU, eventually occurs.

For $\eps_2=-1$ the basic structure is shown in the last column of \Tbl{overview} and in \Fig{ABDiagrams}(b, d).
There are two fixed points before
and after the quadfurcation with positions given by the inner $+$ sign
in \Eq{cZeroFPs}. Using this and \Eq{xi2Ofxi1}, there is no sign choice that smoothly connects the $(\xi_1^*,\xi_2^*)$ branches for $\Delta<0$ to $\Delta > 0$: the fixed points lose their identity when they collide.
The sign choice implies that $\sgn(SC) = -\sgn(\eps_1)$.
When $\eps_1 = 1$, and hence $(A^\txtQ,B^\txtQ) < (4,6)$, the fixed
points both before and after the quadfurcation are below the $SC$ line, so the transition is
\beq{TwoEHtoTwoEH}
	2\mbox{ EH} \to 2 \mbox{ EH}
\eeq
if $|\beta-\gamma|< 2$. As shown in \Fig{ABDiagrams}(b), the fixed points move in towards
the $SC$ line (black and red curves) as $\Delta \to 0^+$ colliding at $\Delta = 0$, and splitting
apart again for $\Delta < 0$.
Similarly, when $|\beta-\gamma|>2$ the transition corresponds to
$$
	2\mbox{ IH} \to 2 \mbox{ IH} .
$$
Finally, when $\eps_1 = -1$ and $(A^\txtQ,B^\txtQ) > (4,6)$, the transition is
\beq{TwoHHtoTwoHH}
	2\mbox{ HH} \to 2\mbox{ HH} ,
\eeq
as shown in \Fig{ABDiagrams}(d).

It is interesting that all of this structure is quite different from what would be expected from a pair of decoupled, area-preserving maps undergoing saddle-center bifurcations, where there can be at most one EE point.
This case corresponds to the special point $(A^\txtQ, B^\txtQ) = (4, 6)$, which will be treated in \Sec{Symmetric} and applied to the case of decoupled maps in \Sec{HenonMaps}.

%%%%%%%
%%%% Collision
%%%%%%%
\subsection{Two to Four Fixed Point Transitions}\label{sec:TwoToFour}
When a parameter path crosses one of the surfaces $b_\pm(a,c)$ then $SC = 0$ in \Eq{SCPD}, and the resulting saddle-center bifurcation typically creates or annihilates a pair of new fixed points, one with E eigenvalues and one with H eigenvalues.
When $\eps_2 = -1$, there are two fixed points outside the wedge between $b_+$ and $b_-$ shown in \Fig{BifSurfaces}(b), and so if the parameter path enters the wedge then two new fixed points are created.
If the path enters the wedge at the quadfurcation point $a=b=c=0$, then the two existing fixed points merge, and the
bifurcation---now a quadfurcation---occurs at the origin.
Depending upon $\beta-\gamma$ and the sign $\eps_1$, the quadfurcation can occur at any point along the $SC$ line, and so a number of different stability cases can arise.

%%%%%
\InsertFig[b]{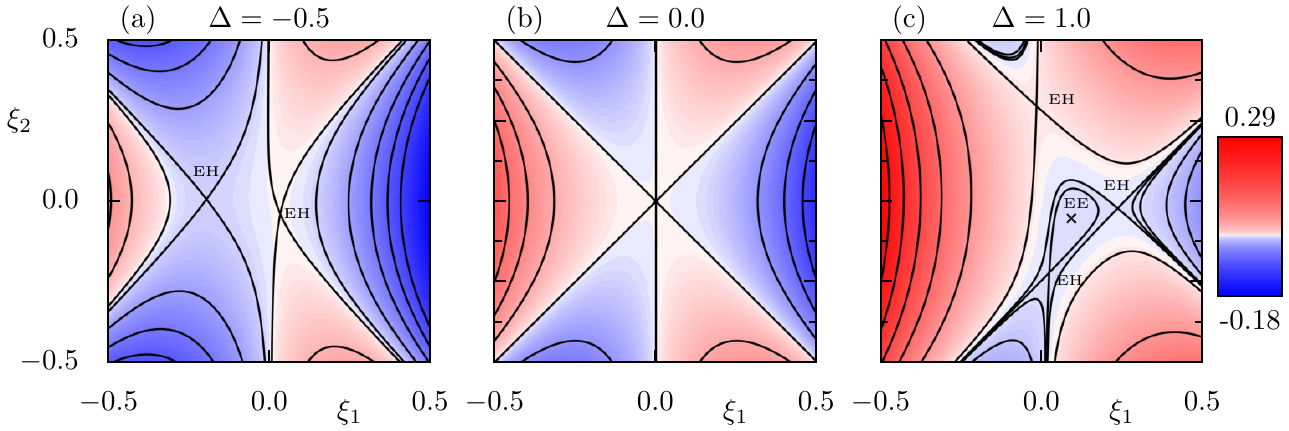}{Contour plots of the
potential $U$, Eq.~\eqref{eq:UPotential},
for the path \eqref{eq:2to4-path} with
$(\alpha, \mu, \delta) = (1, 0.1, 0.5)$, $\eps_1=1$, and $\eps_2 =-1$.
The three panels show $\Delta=-0.5, 0$, and $1$.
Two EH fixed points merge in the quadfurcation at $\Delta=0$ and
for $\Delta>0$ there are four fixed points, one being EE
and three EH.}{potential2to4}{Note:WidthIsDeterminedFromPdfFile}
%%%%%

%%%%%
\InsertFig{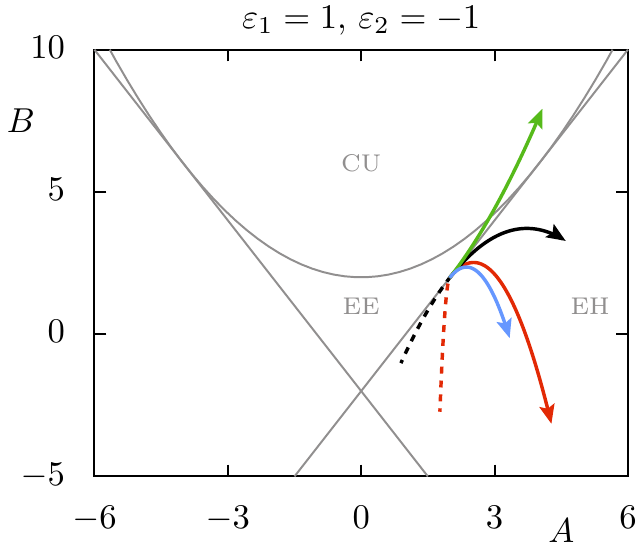}{Stability diagram
for the path \eqref{eq:2to4-path} with $\Delta\in[-2, 5]$
and all other parameters as in \Fig{potential2to4}.
The two EH fixed points for negative $\Delta$ (dashed curves)
merge at the quadfurcation point $\Delta=0$
and lead to four fixed points for $\Delta>0$,
one with EE and three with EH stability. For larger $\Delta$ the EE
point becomes type CU.
}{stability2to4}{Note:WidthIsDeterminedFromPdfFile}
%%%%%

An example is shown in \Fig{potential2to4} for the parabolic path
\begin{equation} \label{eq:2to4-path}
	(a,b,c) = (-0.07 \cdot \Delta |\Delta|, 0.01 \cdot\Delta^2 , \Delta)
\end{equation}
as $\Delta$ varies, with the remaining parameters as shown in the caption.
Figure~\ref{fig:potential2to4} shows the
contours of the potential for $\Delta=-0.5, 0$, and $1$
and \Fig{stability2to4} shows the corresponding stability diagram.
There are two fixed points when $\Delta < 0$ both of type
EH; these merge at $\Delta=0$.
The quadfurcation corresponds to a transition
$$
	2 \mbox{ EH} \to  3 \mbox{ EH } + \mbox{EE} .
$$
Effectively the original EH pair is reformed and the
contour lines near the new EE--EH pair in \Fig{potential2to4}(c)
resemble those for a local saddle-center bifurcation.

%%%%%%%%%%%%%%%%%
%% Symmetric C
%%%%%%%%%%%%%%%%%
\subsection{Krein Collisions, Symmetric $C$ and Reversibility}\label{sec:Symmetric}
As we noted in \Eq{ABQuad}, the quadfurcation occurs for a multiplicity four unit eigenvalue only when the matrix $C$ of \Eq{4DMap} is symmetric. This is the only case in which the quadfurcation can immediately create fixed points of type CU, recall \Fig{ABPlane}.

As discussed in \Sec{SecondDifference} the inverse of the Moser map is conjugate to the original map upon the replacement $C \to C^T$.
Therefore, when $C$ is symmetric, the map \Eq{MoserMap} is reversible, i.e., it is conjugate to its inverse \cite{Lamb98c}:
$S \circ f = f^{-1} \circ S$ for a homeomorphism $S$.
For example, the \Hen~map \Eq{HenonMap} is reversible with $S(x,y) = (y,x)$.
In general, the inverse of \Eq{ShiftedMap} is
$$
	M^{-1}(\xi,\eta) = \left(C^{-1}\eta, \eta + C^T(-\xi + C^{-1}\eta) + \nabla U(C^{-1}\eta) \right) .
$$
This map is conjugate to $M$ when $C=C^T$ using the reversor
$$
	S(\xi,\eta) = (C^{-1}\eta,C\xi) .
$$
Thus, as Moser showed \cite{Moser94}, if $C$ is symmetric the map
\Eq{MoserMap}, or equivalently \Eq{ShiftedMap}, is reversible;
we do not know if the converse of this statement is true.
This reversor is an involution with the fixed set $\mbox{Fix}(S) = \{(\xi,C\xi) : \xi \in \bR^2\}$,
a \twoD{} plane. All of the fixed points are thus symmetric.

%%%%%
\InsertFig[b]{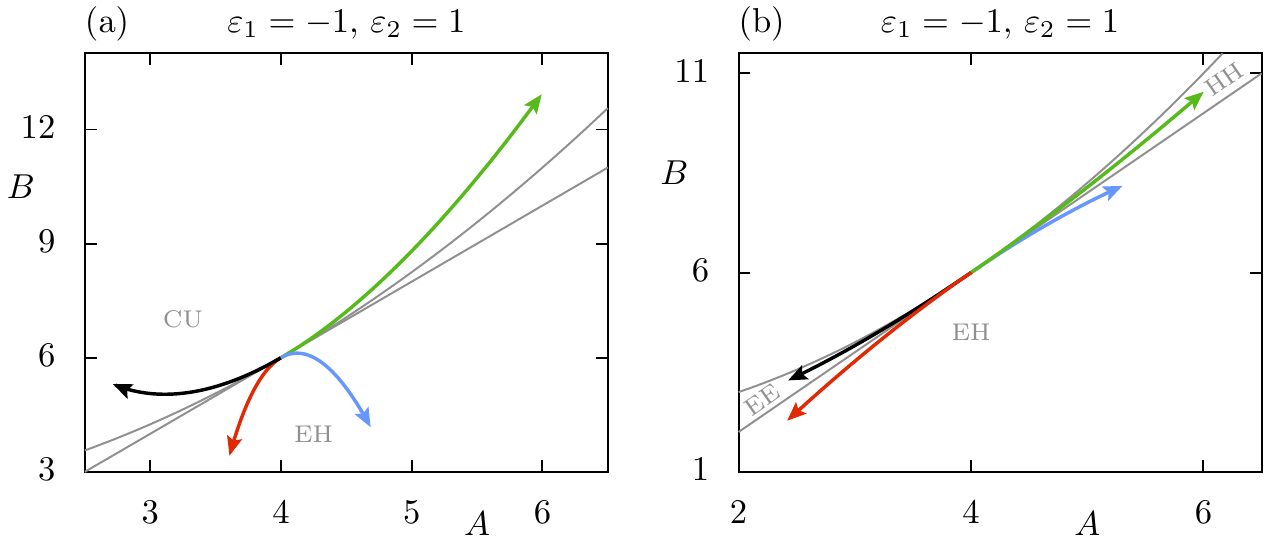}{Stability of fixed points for symmetric $C$ with $\eps_1=-1$ and $\eps_2 = +1$.
The parameters vary along the line $(a,b,c) = \Delta(\tfrac32, \tfrac12,1)$ with $\Delta \in [-1, 0]$.
(a) The transition \Eq{TwoCUTwoEH} for
$(\alpha, \mu, \delta) = (-\tfrac32, 1, \tfrac12)$, giving $\beta=\gamma = \tfrac12$.
(b) The transition \Eq{EEHHTwoEH} for
$(\alpha, \mu, \delta) = (\tfrac32, \sqrt{\tfrac{23}{2}}, \tfrac54)$, giving
$\beta=\gamma = \sqrt{\tfrac{23}{8}}$.
}{ABDiagrams2}{Note:WidthIsDeterminedFromPdfFile}
%%%%%

The quadfurcation is especially interesting in the reversible case since, by \Eq{AB} and \Eq{OrderDeltaStab}, only
then does it occur for the stability parameters $(A^\txtQ,B^\txtQ) = (4,6)$. Indeed as in \Sec{QuadLines}, assuming that $a,b,c = \cO(\Delta)$ and $\xi_i= \cO(\sqrt{\Delta})$ near the quadfurcation, the Krein criterion \Eq{OrderDeltaStab-K} is zero to $\cO(\Delta^{1/2})$, and the first nonzero terms are
$$
	KP = 4\alpha \delta a -2(\alpha+3\eps_2\delta)\gamma b
		  +\left(12\eps_1\eps_2 - (\alpha-3\eps_2\delta)^2\right)\xi_1^{*2}
		  + \cO\left(\Delta^{3/2}, (\beta-\gamma)^2\right)
$$
where $\xi_1^*$ is given by \Eq{cZeroFPs}.
The sign of this parameter can change, depending upon
the details. However, note that since $KP$ depends only on $\xi_1^{*2}$, it does not depend upon the outer
sign in \Eq{cZeroFPs}. Thus when $c = \cO(\Delta)$, the fixed points come in pairs with the same sign of $KP$.

When $\eps_2 = 1$, the four created fixed points come in pairs with opposite signs of $SC$ from
\Eq{OrderDeltaStab-SC}. Thus there will be a pair of fixed points of type EH. Since the curves
generically emerge tangent to the the $SC$ line,
the second pair will both have type CU or one will be  EE and the other HH.
For example, a quadfurcation of the form
\beq{TwoCUTwoEH}
	\emptyset \to 2 \mbox{ CU} + 2 \mbox{ EH}
\eeq
is shown in \Fig{ABDiagrams2}(a).
For this case $KP = \frac14(3\pm7\sqrt{6})\Delta$ to lowest order, where the sign is the inner $\pm$ sign in \Eq{cZeroFPs}, implying that the two fixed points with the $+$ sign have $KP > 0$ when $\Delta < 0$ near the quadfurcation and are thus of type CU.
In contrast, for the example shown in \Fig{ABDiagrams2}(b)
the quadfurcation is
\beq{EEHHTwoEH}
	\emptyset \to \mbox{ EE} + \mbox{ HH} + 2\mbox{ EH}
\eeq
as here $KP = \frac{\Delta}{64} (993 - 84 \sqrt{46} \pm 91 \sqrt{6}) < 0$ when the fixed points exist, $\Delta<0$.
This last case is what would happen in a pair of uncoupled \twoD{} maps, and will be seen below in \Sec{HenonMaps}.

As in \Sec{QuadLines}, when $\eps_2 = -1$, the quadfurcation at $\Delta = 0$ corresponds to a collision and re-emergence of a pair of fixed points with $\sgn(SC) = -\sgn(\eps_1)$. When $\eps_1 = 1$, all of the fixed points will have stability type EH, and so the transition will be
$$
   2 \mbox{ EH} \to 2 \mbox{ EH}
$$
and thus follows the pattern shown in \Fig{ABDiagrams}(b).

However, when $\eps_1 = -1$, then $\sgn(SC) = +1$ implying that EE, HH  and CU are all possible. However, since the two fixed points have the $+$ inner sign in \Eq{cZeroFPs} they will have the same sign of $KP$, so we either have a CU pair or an EE+HH pair. The possible transitions are
\beq{TwoCUtoEEHH}
	2 \mbox{ CU} \to \mbox{ EE} + \mbox{ HH},
\eeq
for which the stability diagram is shown in \Fig{ABDiagrams3}(a)
\beq{EEHHtoEEHH}
	\mbox{ EE} + \mbox{ HH} \to \mbox{ EE} + \mbox{ HH},
\eeq
with stability diagram shown in in \Fig{ABDiagrams3}(b)
and
\beq{TwoCUtoTwoCU}
	2 \mbox{ CU}  \to 2 \mbox{ CU},
\eeq
with stability diagram shown in \Fig{ABDiagrams3}(c).

%%%%
\InsertFig[b]{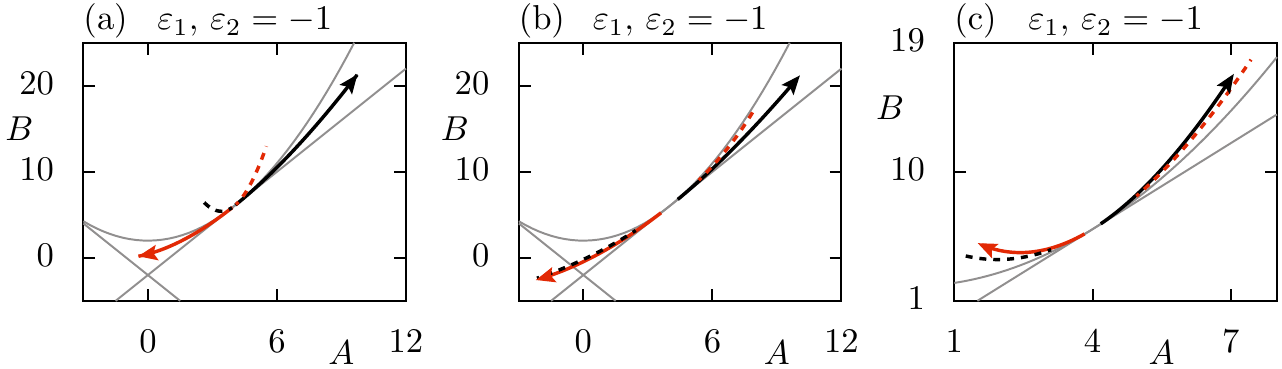}{Stability of fixed points for symmetric $C$ and $\eps_1= \eps_2=-1$.
The parameters $(a,b,c) = \Delta(\tfrac32,\tfrac12,1)$ for (a) and (b), and $\Delta(2,\tfrac12,1)$ for (c),
with $\Delta \in [-\tfrac12, \tfrac12]$. The arrows indicate the direction towards more negative $\Delta$, and
branches for $\Delta>0$ are shown as dashed lines.
(a) The transition \Eq{TwoCUtoEEHH} for $(\alpha,\mu,\delta)=(\tfrac12,\sqrt{6},1)$ giving $\beta=\gamma = \sqrt{\tfrac32}$.
(b) The transition \Eq{EEHHtoEEHH} for $(\alpha,\mu,\delta)=(\tfrac{99}{16}, -\tfrac52, \tfrac{1}{11})$ giving $\beta=\gamma = -\tfrac54$.
(c) The transition \Eq{TwoCUtoTwoCU} for $(\alpha,\mu,\delta)=(-\tfrac32, 1, \tfrac12)$, giving $\beta=\gamma = \tfrac12$.
}{ABDiagrams3}{Note:WidthIsDeterminedFromPdfFile}
%%%%

%%%%%
\InsertFig[t]{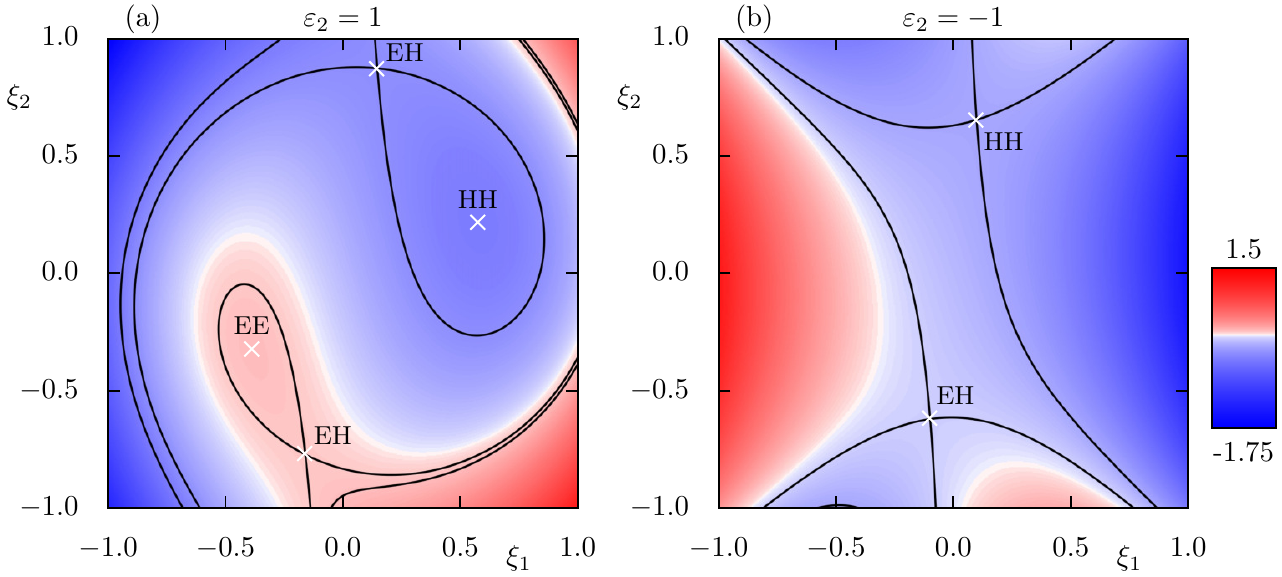}{
Contour plot of the potential $U$, Eq.~\Eq{UPotential}. (a) For $(a,b,c) = (-\tfrac34,-\tfrac14,-\tfrac12)$ and $\eps_2=1$
there are four critical points, which are equilibria of \Eq{HamSymmetric}.
When $C >0$, the maximum, at
$\xi = (-0.38825, -0.32196)$, corresponds to a doubly elliptic equilibrium
and the minimum,
at $\xi = (0.57441,  0.21761)$, to a doubly hyperbolic equilibrium.
As specific example $(\alpha, \mu, \delta) = (\tfrac52, \sqrt{6}, 1)$
is used for the matrix $C$ determining the stabilities.
(b) Potential for $(a, b, c)=(-\tfrac38, -\tfrac18, -\tfrac14)$
and $\eps_2=-1$ where there are two critical points; the matrix $C$, determining the stabilities, is the same as in (a).}{Potential}{Note:WidthIsDeterminedFromPdfFile}
%%%%%

Another technique for analyzing stability in the neighborhood of the quadfurcation is to use
the ODE limit \Eq{ODEs} with the Hamiltonian \Eq{Hamiltonian}.
Recall that this limit assumes that we assume the scaling $(a,b,c) \to (h^4a,h^4 b,h^2c)$,
for $h \ll 1$. This scaling differs from the $\Delta$-scaling by allowing larger relative values for $c$.
The implication of this is that the term $\tfrac12  c q_1^2$ in $U$, which was negligible when
$c = \cO(\Delta)$, is formally important when we take $c = \cO(h^2)$. When $C$ is symmetric,
$\Ca = 0$, and the Coriolis and centripetal-like terms vanish, simply giving
\beq{HamSymmetric}
	H(q,p) = \tfrac12 p^T C^{-1} p - U(q) .
\eeq

When $C = \Cs$ is positive or negative definite, then the stability is governed entirely by the classification of
the critical point of $U$. In particular if $C>0$ then since the potential in \Eq{Hamiltonian} is $-U$, a minimum of $U$ is an HH point, and a maximum is an EE point. Saddles, correspond to EH points. When $C < 0$, the minima are HH and the maxima are EE. This is consistent more generally with \Eq{SCPD}, which shows that $SC = \eps_1 \det(D^2 U)$. Two example contour plots for $U$ are shown in \Fig{Potential}.

%%%%%%%%%%%%%%%%%
%% Visualization
%%%%%%%%%%%%%%%%%
\subsection{Elliptic Bubbles}\label{sec:Bubbles}

To visualize the dynamics near the fixed points
of the \fourD{} map
we use a \textit{\threeD{} phase space slice} \cite{RicLanBaeKet2014}.
In its simplest form one considers a thickened
\threeD{} hyperplane  in the \fourD{} phase space defined
by fixing one of the coordinates, e.g., $\eta_2=\eta_2^{\star}$,
to define the slice of thickness $\epsilon$ by
$$
     \left\{ (\xi_1, \xi_2, \eta_1, \eta_2) \; \left| \rule{0pt}{2 ex} \;
         |\eta_2 - \eta_2^*| \le \epsilon \right. \right\},
$$
Whenever the points of an orbit lie within the slice, the remaining
coordinates $(\xi_1, \xi_2, \eta_1)$ are displayed in a \threeD{} plot.
The parameter $\epsilon$ determines the resolution of the resulting plot;
decreasing $\epsilon$ requires the computation of longer trajectories
as the slice condition is fulfilled less often, but the resulting intersections will be more precise.
For example, if a two-torus intersects the hyperplane, it will typically do so in one or more loops. As
$\epsilon$ grows these loops thicken into annuli in the slice.
For further examples and detailed discussion see
\cite{RicLanBaeKet2014,LanRicOnkBaeKet2014,OnkLanKetBae2016,Lange16,AnaBouBae2017,FirLanKetBae2018:p}.

%%%%%%%%%%%
\InsertFig[t]{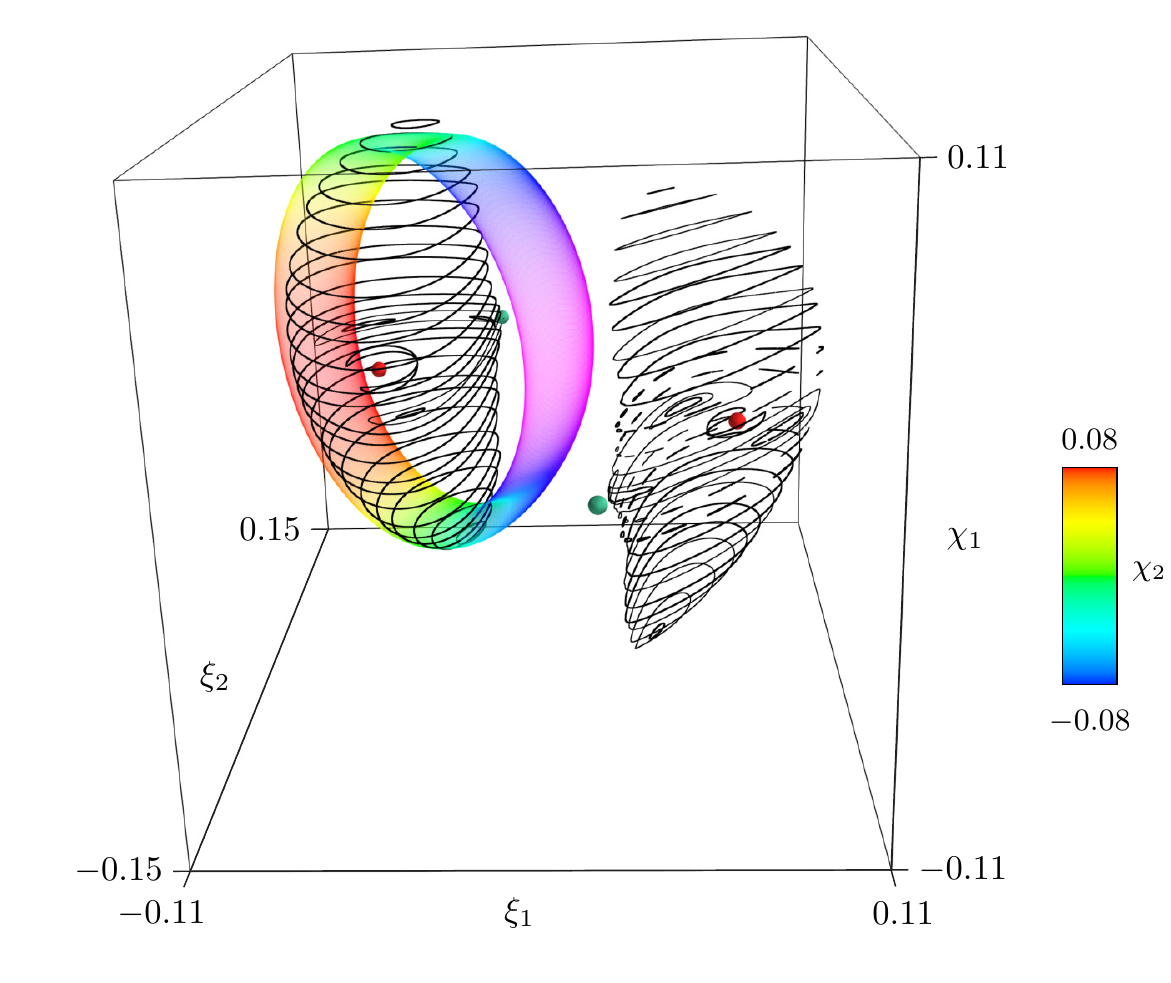}{Three-dimensional phase space slice of the \fourD{} map
             corresponding to  \Fig{ABDiagrams}(a) with parameters $(\alpha, \mu, \delta)=(1, 0.1, 0.5)$,
             $\eps_1=\eps_2=1$, and $(a, b, c)=(-0.015, -0.005, -0.01)$.
             The small spheres show two EE (red) and two EH (green) fixed points.
             Also shown are several selected regular tori (black lines)
             surrounding the EE fixed points.
             Each torus is represented by $10^4$ points in the slice with $\epsilon=10^{-6}$.
             These tori are \twoD{} in the \fourD{} phase space
             and therefore (usually) lead to a pair of loops in the
             \threeD{} phase space slice.
              A projection of one \twoD{} torus is shown as $10^6$ semi-transparent
              points with $\chi_2$ encoded in color (see color bar).
             \movierefall
             }{3d-pss}{WidthIsDeterminedFromPdfFile}
%%%%%%%%%%%

For our purposes a slightly more general, rotated slice, defined so as to contain
all of the fixed points, will be more convenient.
Because the momenta of the fixed points are determined by the coordinates through $\eta=C\xi$,
all fixed points of the \fourD{} map \Eq{ShiftedMap} are
contained in a \twoD{} plane.
Following the ideas of \cite[App.~3]{RicLanBaeKet2014},
we define new coordinates $(\xi,\chi)$, with $\chi = \eta-C\xi$,
so that the fixed points lie in the two-plane $\chi = 0$.
Thus we define the \threeD{} slice
$$
  \Gamma_{\epsilon} =
     \left\{ (\xi_1, \xi_2, \chi_1, \chi_2) \; \left| \rule{0pt}{2 ex} \;
         |\chi_2| \le \epsilon \right. \right\}
$$
so that we get $(\xi_1, \xi_2, \chi_1)$ as \threeD{} coordinates.
Equivalently, this corresponds to using non-orthogonal basis vectors
given as columns of the block matrix
$$
	B = \left( \uv_1 | \uv_2 | \uv_3 | \uv_4 \right)
	  = \begin{pmatrix} I & 0  \\
						C & I
		\end{pmatrix}.
$$
As these are linearly independent, they
can be used to express any point as
linear combination with coefficients
$(\xi_1, \xi_2, \chi_1, \chi_2)$. These coefficients
can be computed from the scalar products with the
dual basis vectors $\{\uv^i\}$, which are the columns of
$$
	B^{-T} = \left( \uv^1 | \uv^2 | \uv^3 | \uv^4 \right)
	  = \begin{pmatrix} I & -C^{T}  \\
						0 & I
		\end{pmatrix}.
$$

Figure \ref{fig:3d-pss} shows an example of a slice for the map \Eq{ShiftedMap}
with the parameters of \Fig{ABDiagrams}(a) when $\Delta = -0.01$.
For these parameters the quadfurcation has created
two EE and two EH fixed points that, by construction
of the \threeD{} slice, lie in the \twoD{} plane $\chi_1=0$.
As expected from KAM theory, the EE fixed points should be surrounded by a Cantor family of two-tori
on which the dynamics is conjugate to incommensurate rotation. By analogy with
Moser's theorem for \twoD{} maps \cite{Moser62}, the density of these tori should approach
one as they limit on the EE points providing that the linearized frequencies
are not in a low-order resonance.
Indeed, the formal normal form expansion
around a nonresonant EE point is an integrable twist map to all orders \cite{Todesco94}, and
higher-dimensional results along the lines of Moser's twist theorem
have been proven for elliptic equilibria of Hamiltonian flows \cite{DelGut1996, Eliasson13}.
In the \threeD{} slice each of these tori becomes
two (or more) thin annular rings, which, since we have set $\epsilon = 10^{-6}$,
appear as \oneD{} loops in the figure \cite{RicLanBaeKet2014}.

This family of two-tori appears to approximately be limited by the locations of the EH fixed points.
Indeed, since the center-stable and center-unstable manifolds of the EH points are \threeD{},
they should form boundaries for the elliptic dynamics. Of course, we expect
there will be chaotic orbits near these manifolds, and so the
regular \twoD{} tori will not extend into the chaotic zone.
The black loops in the plot are a selection of \twoD{} tori
that are close to the boundary of the regular region.
Also shown in the plot is the full orbit of one of these \twoD{} tori, now \textit{projected} onto the slice;
the projected coordinate $\chi_2$ is
encoded in color as indicated in the color bar at the right \cite{PatZac1994}.

%%%%%%%%%%%
\InsertFig[t]{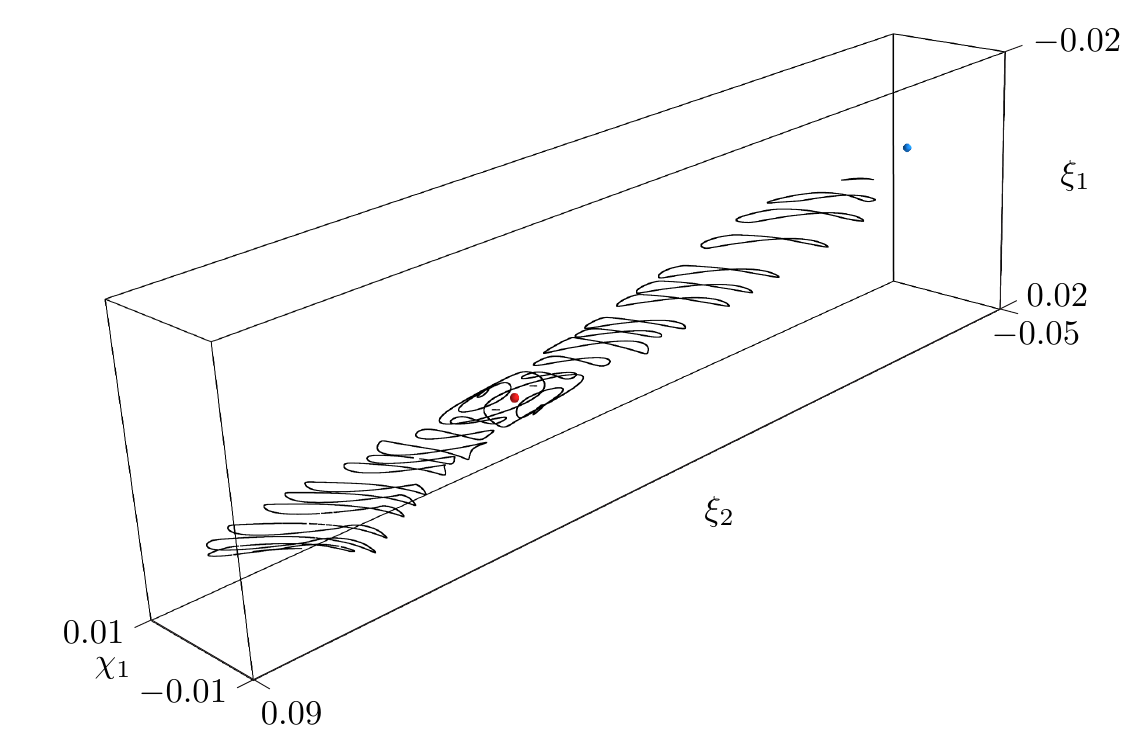}{\threeD{} phase space slice of the \fourD{} map
             for parameters of \Fig{ABDiagrams3}(b) with $\Delta = -0.001$.
             The small spheres show an EE (red) and an HH (blue)
             fixed point.
             Also shown are slices of several selected regular tori as in \Fig{3d-pss}.
             \movierefall
             }{3d-pss-2}{WidthIsDeterminedFromPdfFile}
%%%%%%%%%%%

Figure~\ref{fig:3d-pss-2} shows an example of a slice for the
transition \Eq{EEHHtoEEHH}, when the matrix $C$ is symmetric; the parameters
correspond to \Fig{ABDiagrams3}(b) with $\Delta = -0.001$,
close to the quadfurcation.
For this case there are only two fixed points, one of type EE and the other of type HH.
Here again we see a family of \twoD{} tori surrounding the EE point.
This family has a larger extent in the $\xi_2$ direction than in $\xi_1$,
and the loops shrink in size as they become closer to the HH point.
Note that now the stable and unstable manifolds
of the HH point are two-dimensional, and so do not form barriers in \fourD{}.

%%%%%%%%%%%%%%%%%
%% Bounded Orbits
%%%%%%%%%%%%%%%%%
\subsection{Bounded Orbits }\label{sec:Bounded}
Moser showed, under a nondegeneracy condition on the quadratic terms, that
the domain of the quadratic map containing bounded orbits is itself bounded \cite{Moser94}.
To obtain an explicit bound we consider the second-difference form \Eq{SecondDiff1}, rewriting it as
\beq{OrderedDiff}
	C^T \xi_{t+1} + C \xi_{t-1} =  A + D \xi_t + Q(\xi_t),
\eeq
where $A = (a,b)^T$ is a constant vector and the linear and quadratic terms are
\begin{align*}
	D \xi &\equiv \begin{pmatrix} 2\alpha + c & \mu \\ \mu & 2\delta \end{pmatrix}
				\begin{pmatrix} \xi_1 \\ \xi_2 \end{pmatrix} ,\\
	Q(\xi) &\equiv \begin{pmatrix} 3\eps_2 \xi_1^2 + \xi_2^2 \\ 2 \xi_1 \xi_2 \end{pmatrix}.
\end{align*}
 Using this form we can prove the following.

%%%%%%%%%
\begin{thm}\label{thm:Bounded}
When $\eps_2 = \pm1$, all bounded orbits of the map \Eq{OrderedDiff} are contained in
the disk $\|\xi\| \le \kappa$, where
\beq{K1Define}
	\kappa \equiv \tfrac{1}{2\tau} \left(\kappa_2 + 2\kappa_3  + \sqrt{ (\kappa_2+2\kappa_3)^2 +4\tau \sqrt{a^2+b^2}} \right),
	\\\
\eeq
with $\tau = 1$ for $\eps_2 = 1$ and $\tau = \sqrt{\tfrac23}$ for $\eps_2 = -1$,
and we define
\bsplit{K23Define}
	\kappa_2 &\equiv \|D\|_F = \sqrt{(2\alpha+c)^2 + 2\mu^2 + 4\delta^2} , \\
	\kappa_3 &\equiv \|C\|_F = \sqrt{\alpha^2 + \beta^2 + \gamma^2 + \delta^2 }.
\esplit
\end{thm}
\begin{proof}
When $\eps_2 =1$ the norm of the quadratic terms has the lower bound
$$
	\|Q(\xi)\|^2 = 9\xi_1^4 + 10 \xi_1^2 \xi_2^2 + \xi_2^4 = (9\xi_1^2 + \xi_2^2)\|\xi\|^2
				\ge \|\xi\|^4 = \rho^4 ,
$$
where we denote $\|\xi\| = \rho$.

Using the triangle inequality on \Eq{OrderedDiff} gives
\beq{Triangle1}
 \|C^T \xi_{t+1}\| + \|C \xi_{t-1}\| \ge  \|Q(\xi_t)\| - \|D \xi_t \| - \|A\| .
\eeq
Define $\kappa_{2,3} >0 $ so that
$$
	\|D \xi\| \le \kappa_2 \rho ; \quad
	\|C^T \xi\|, \|C \xi\| \le \kappa_3 \rho .
$$
For example we can use the Frobenius norms of these matrices to give \Eq{K23Define}
(or the operator norm, in terms of the singular values).
Putting the bounds into \Eq{Triangle1} gives
\beq{Triangle}
 \kappa_3 ( \rho_{t+1} + \rho_{t-1} )  \ge \rho_t^2- \kappa_2 \rho_t - \|A\| .
\eeq
Let $\kappa>0$ be chosen such that whenever $\rho > \kappa$, then
$ \rho^2  -\kappa_2 \rho  -\|A\| > 2\kappa_3\rho$. Solving this quadratic, as an equality, gives \Eq{K1Define}.
Using this in \Eq{Triangle} implies that whenever $\rho_t > \kappa$,
\beq{DiffIneq}
	\rho_{t+1} + \rho_{t-1} > 2\rho_t .
\eeq
Now there are two possible cases:
\begin{itemize}
	\item Suppose that $\rho_t \ge \rho_{t-1}$. Then by \Eq{DiffIneq}
whenever $\rho_t > \kappa$, we have $\rho_{t+1} > 2\rho_t - \rho_{t-1} \ge \rho_t$. This implies that
the sequence $\{\rho_t\}$ is strictly increasing with $t$. If this monotone sequence
is bounded, it must approach a limit $\rho_t \to \rho^*$, which must be a solution of \Eq{Triangle}
as an equality. But this implies $\rho^* = \kappa$, and we have assumed $\rho_t > \kappa$. Thus $\{\rho_t\}$
is unbounded as $t \to \infty$.

	\item Alternatively, suppose that $\rho_t < \rho_{t-1}$. Again, by
\Eq{DiffIneq} whenever $\rho_{t+1} > \kappa$ we have $\rho_{t-1} > 2\rho_t - \rho_{t+1} > \rho_{t}$.
Thus the sequence $\{\rho_t\}$ strictly increases as $t$ decreases. Again, this implies that $\rho_t$
is unbounded, now as $t \to -\infty$.
\end{itemize}
Together, these imply the theorem when $\eps_2 = 1$.

If $\eps_2 = -1$, then we can see that $\|Q(\xi)\|^2 \ge \tfrac23 \rho^4$, indeed
$$
	\|Q(\xi)\|^2-\tfrac23 (\xi_1^2+\xi_2^2)^2 = \tfrac{25}{3}\xi_1^4 -\tfrac{10}{3}\xi_1^2\xi_2^2 + \tfrac{1}{3} \xi_2^4 = \tfrac13( 5\xi_1^2 -\xi_2^2)^2 \ge 0 .
$$
Thus the analysis above works, if we replace $\kappa$ by the larger solution to
$$
	\tau\rho^2 -(\kappa_2+2\kappa_3)\rho - \|A\| = 0
$$
with $\tau = \sqrt{\tfrac23}$, giving \Eq{K1Define} again.
\end{proof}

Note that if $\eps_2 = 0$, then
$$
	\|Q(\xi)\|^2 = 4\xi_1^2 \xi_2^2 + \xi_2^4 = \xi_2^2(4\xi_1^2 + \xi_2^2) ,
$$
which does not obey a bound of the form needed in the proof of \Th{Bounded}. Thus the theorem does not apply to this case. Indeed, we showed in \Sec{FixedPoints} that if $\eps_2 =0$, then when $a=b=c=0$ there is a line of fixed points, recall the discussion in \Sec{FixedPoints}.

%%%%%%%%%%%%%
\InsertFig[t]{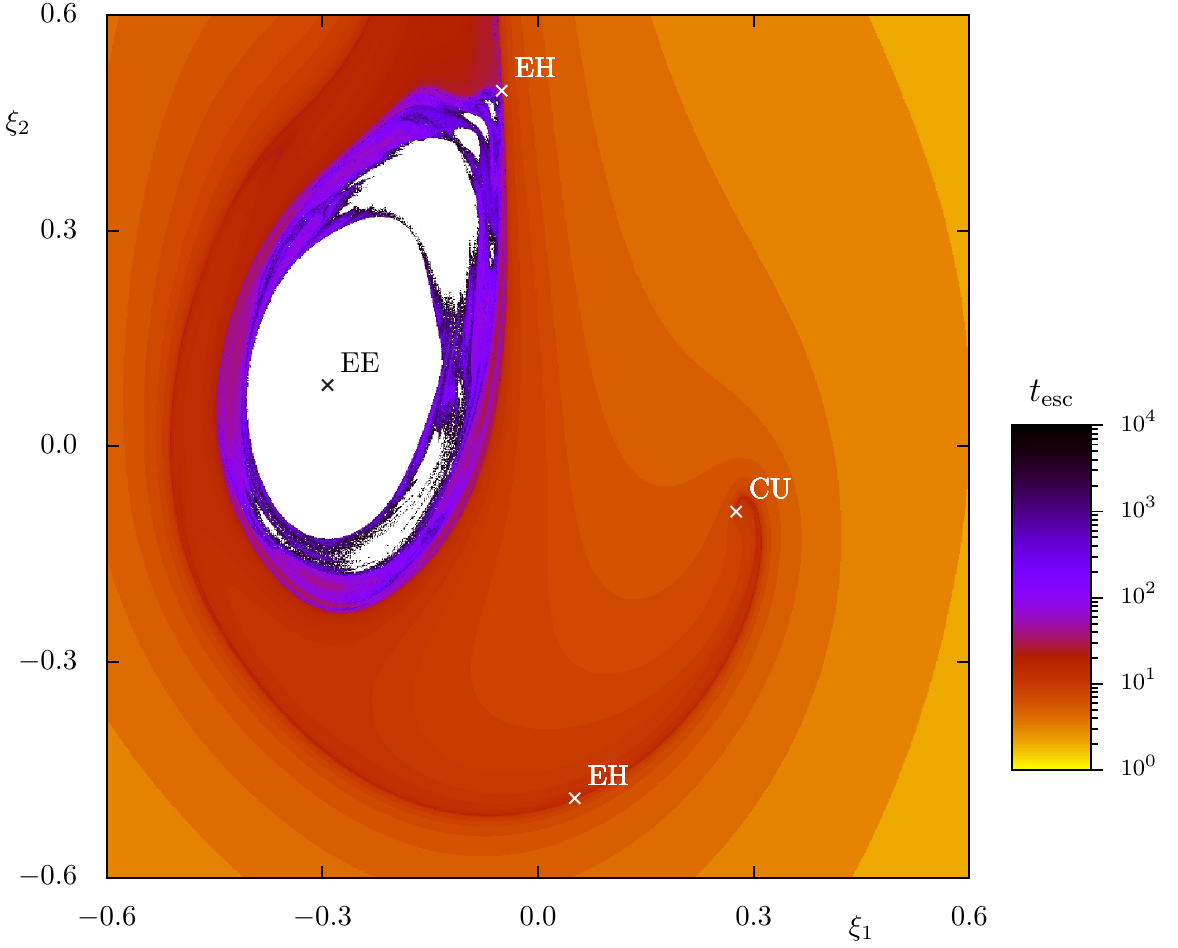}{Plot of the escape time
$\tesc$ encoded in color for initial conditions
defined via $(\xi_1, \xi_2)$ with $\eta = C\xi$. Initial conditions whose orbits have
not escaped within $10^4$ iterations are colored in white.
In the region surrounding the EE fixed point one has a large
region of non-escaping orbits. Further away
one observes a complicated fine-structure
of escaping and non-escaping orbits.
Parameters are: $(\alpha, \mu, \delta)=(1, 0.1, 0.5)$,
$\eps_1=\eps_2=1$, and
$(a, b, c) = (-0.25, 0.05, 0.05)$.
}
{EscapeTimePlot}{WidthIsDeterminedFromPdfFile}
%%%%%%%%%%%

A good way to visualize the distinction between bounded
and unbounded orbits for a given parameter set
is an escape time plot, see \Fig{EscapeTimePlot}.
In this plot, a grid of initial points of the form $(\xi,\eta) = (\xi, C\xi)$
are iterated until $\| \xi \| > \kappa$
and the required time to escape is encoded in color.
Points that have not escaped within $10^4$ iterations are displayed in white.
Some of these points lie on
a family of regular \twoD{} tori in the neighborhood of the EE fixed point; these will never escape.
Points near the boundary of the white region
may eventually escape, and indeed, even arbitrarily close to an EE fixed point there are
initial conditions that are expected to escape for extremely large times
by means of Arnold's exponentially-slow diffusion mechanism \cite{Arnold64,Chi1979,Lochak99,Dum2014}.

We can quantify the size of the region of ``bounded orbits" under parameter
variation by computing the area of the white region
in a \twoD-plane of initial conditions like that in \Fig{EscapeTimePlot}.
For this we choose initial conditions in the two-plane $(\xi,C\xi)$,
with $\xi$ varied on a grid of $3000\times3000$ points within the box
$ -\tfrac12 \kappa < \xi_i < \tfrac12 \kappa$, and iterate at most $5000$ steps. Orbits that remain within the disk
$\| \xi \| \le \kappa$ are counted, and the resulting area is denoted $\Areg$.
This area varies as fixed points undergo various bifurcations.
A \oneD{} cut through parameter space,
\Fig{BoundedOrbitsCut}, shows how $\Areg$ varies with the parameter $a$ (here $b=c=0.01$).
Insets in the figure show escape time plots,
like that in \Fig{EscapeTimePlot}, for some selected parameters.
There is a strong correlation of the area with the structure of the region of
stable orbits around the EE point, as we discuss further below.

%%%%%
\InsertFig[t]{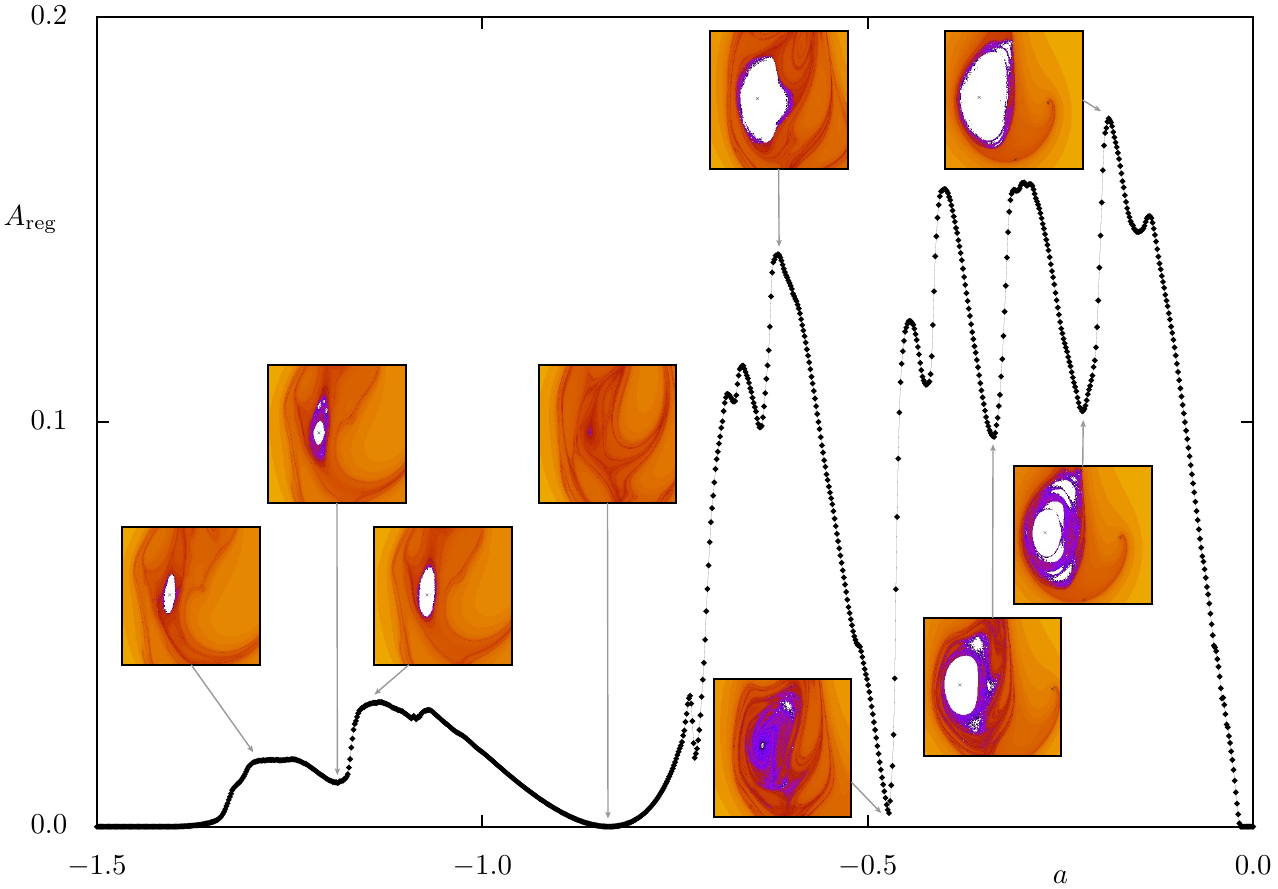}
{Area of bounded initial conditions $(\xi,C\xi)$, determined
from a grid in the $\xi$-plane for $|\xi_i| < \kappa/2$,
under variation of $a$ with $b=c=0.01$, and $\eps_2 = 1$;
the matrix $C$ is as in \Fig{EscapeTimePlot}.
The insets show escape time plots for a neighborhood
of the EE fixed point. Note that for all $a$-values in the figure
there are four fixed points (shown as
small black crosses in the insets), which, as $a$ decreases, move out of the shown square.
A movie of escape time plots as a function of $a$
can be found in the supplementary material at
\href{http://www.comp-phys.tu-dresden.de/supp/}{http://www.comp-phys.tu-dresden.de/supp/}.}
{BoundedOrbitsCut}{WidthIsDeterminedFromPdfFile}
%%%%%

In \Fig{BoundedOrbitsNew} we show $\Areg$ as a function of the two parameters $a$ and $b$, setting $c=b$.
Note that there are apparently no bounded orbits when $b > b_+$ \Eq{bPlusMinus}, where there are no fixed points,
nor when $b \gtrsim 0.6$ or $a < -1.5$. The largest bounded area occurs in the region near the origin
in the $(a,b)$ plane; recall that the origin corresponds
to the quadfurcation since $c=b$.
The dotted curve corresponds to parameters
for which $PD=0$, the period-doubling bifurcation \Eq{OrderDeltaStab-PD}.
To the left of the $PD$ curve, the EE fixed point becomes IE and
$\Areg$ decreases quickly to 0.
To the right of the $PD$ curve and for $b< b_-$,
the four fixed points have the stabilities EE, EH, EH, and CU.
On the curve $b=b_-$ the EH and CU fixed points coalesce; however, since there are no bounded,
regular orbits in a neighborhood these fixed points, this transition does not influence $\Areg$.

%%%%%
\InsertFig[t]{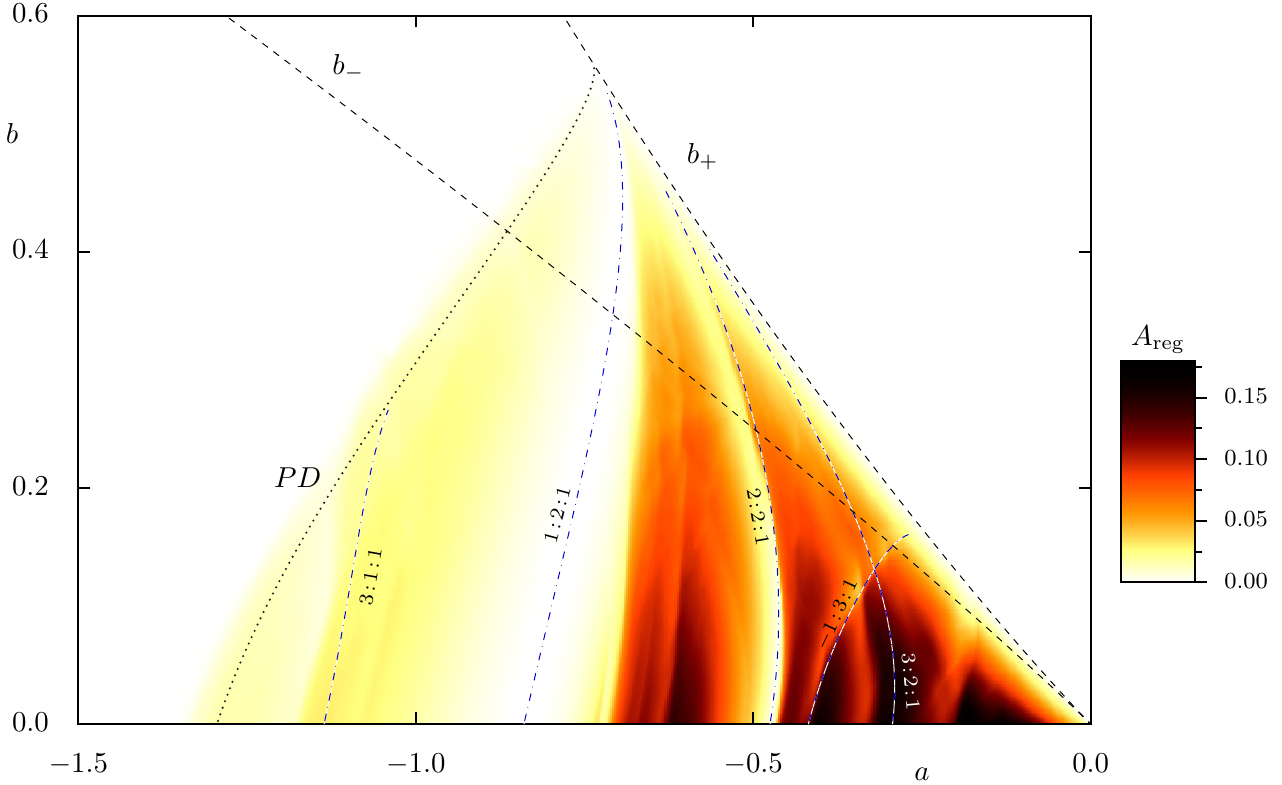}
 {Area $\Areg$
  of bounded initial conditions as a function of $(a,b)$ with $c = b$ and $\eps_2 = 1$.
  The matrix $C$ is defined by $(\alpha, \mu, \delta) = (1, 0.1, 0.5)$, and
  $\eps_1=1$, as in \Fig{ABDiagrams}(a, b).
  The curves $b_\pm$ represent the crossing of the surfaces
  \Eq{bPlusMinus} with the plane $c = b$.
  Thus there are four fixed points when $b$ is below $b_-$ and
  two when $b_- < b <b_+$.
  Upon crossing the (dotted) $PD$ curve from right to left, the EE fixed point becomes
  IE and $\Areg$ rapidly drops to zero.
  The blue dashed-dotted curves show several resonances \Eq{Resonance}, labeled $n_1:n_2:m$, of the EE fixed point.}{BoundedOrbitsNew}{Note:WidthIsDeterminedFromPdfFile}
%%%%%

At several places in \Fig{BoundedOrbitsCut}
and along several curves in \Fig{BoundedOrbitsNew} one observes a substantial decrease in the bounded area.
Several of these can be related to those parameters
for which the linearization about the elliptic-elliptic
fixed point fulfills a low-order resonance.
For such a point the four eigenvalues have the form $\lambda_{1,2} = e^{2\pi i \nu_{1,2}}$, and
the conjugate/inverse values $\bar{\lambda}_{1,2} = 1/\lambda_{1,2}$.
Each frequency, $\nu_{1,2}$, describes the rate of rotation around the
fixed point in the \twoD{} invariant planes
spanned by the eigenvectors of the corresponding conjugate pair of eigenvalues.
These can be written in terms of the partial traces, $\rho_{1,2}$, recall \Eq{PartialTraces}, as
$$
  \nu_{1,2} = \tfrac{1}{2\pi} \arccos\left(\tfrac12 \rho_{1,2} \right) .
$$
The frequencies $(\nu_1, \nu_2)$ of an EE fixed point fulfill
a resonance condition when
\begin{equation} \label{eq:Resonance}
     n_1 \nu_1 + n_2 \nu_2 = m, \quad (n_1, n_2) \in \bZ^2\backslash \{0, 0\}, \,\, m\in\bZ .
\end{equation}
Without loss of generality, we can set  $\text{gcd}(n_1, n_2, m)=1$, and $m \ge 0$. We refer to such
a resonance as an $n_1 : n_2 : m$ resonance.

While resonances are dense in frequency space, those with small order, $|n_1| + |n_2|$,
are of particular relevance. For example the large white region in \Fig{BoundedOrbitsNew} starting for $b=0$
at $a \approx -0.75$ corresponds to the $1:2:1$ resonance.
This is also manifested in the broad minimum with $\Areg\approx 0$
in \Fig{BoundedOrbitsCut}.
The other prominent minimum near $a \approx -0.5$
is caused by the $2:2:1$ resonance.
In these two cases $\Areg$ is reduced for
parameters near those fulfilling the resonance condition.
For other indicated resonances, i.e., $3:2:1$, $-1:3:1$
and $3:1:1$,
the density is only reduced on one side of the bifurcation.
In the examples this happens for smaller $a$, and sometimes---as for the
$3:2:1$ resonance---it occurs quite some distance away.

Higher order resonances, $|n_1|+| n_2| \ge 5$, should be less important in changing $\Areg$.
The results of \cite{Todesco94} lead to the expectation that the EE point remains stable, and that
for a single resonance the bifurcation creates a pair of invariant \oneD{} tori,
one normally hyperbolic and one normally elliptic (at least in the normal form).
Further away from the bifurcation of the EE fixed point the
geometry is described by bifurcations of families of \oneD{} tori,
see \cite{OnkLanKetBae2016} and references therein.
When the frequency passes through a double resonance, so that $\nu_i = p_i/q_i$ are rational, then one expects four periodic orbits to be created \cite{Todesco94, Gelfreich13}.
According to \cite{Todesco94} their stability is either
EE $+$ 2 EH $+$ HH or 2 EH $+$ 2 CU, see \cite{LanRicOnkBaeKet2014} for an illustration
of the geometry in the first case.

%%%%%
\InsertFig[t]{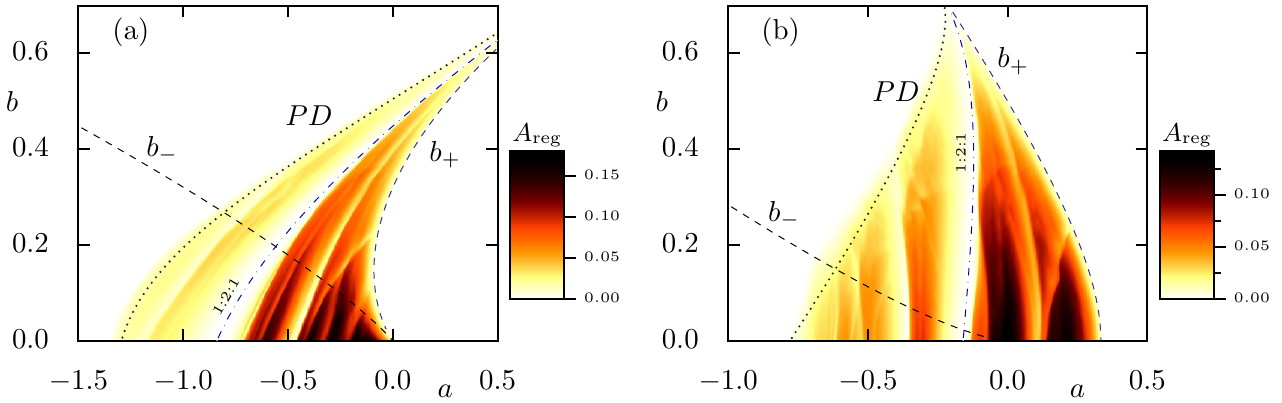}
{Area $\Areg$ of initial conditions on the \twoD{} plane $\eta = C\xi$
inside the box $|\xi_i| \le \kappa/2$ (on a $3000 \times 3000$ point grid)
that remain within the disk $\|\xi_t\| \le \kappa$ for $t\le 5000$.
Variation of $(a,b)$ with (a) $c = 5b$
and (b) $c=2$.
The remaining parameters are the same as in \Fig{BoundedOrbitsNew}.
The curves $b_\pm$ represent the crossing of the surfaces \Eq{bPlusMinus} with the plane $c = b$ or $5b$, respectively.
Thus there are four fixed points when
$b$ is below $b_-$ and two when $b$ is between $b_-$ and $b_+$.
To the left of the (dotted) $PD$ line the EE fixed point
has become EI.}{BoundedOrbits}{Note:WidthIsDeterminedFromPdfFile}
%%%%%

Two further examples of $\Areg$
are shown in \Fig{BoundedOrbits} with the same parameters as
in \Fig{BoundedOrbitsNew}, except for \Fig{BoundedOrbits}(a) $c=5b$
and for \Fig{BoundedOrbits}(b) the parameter $c=2$ is fixed.
In the latter case the parameter plane no longer intersects the
quadfurcation point ($a=b=c=0$),
so that the curves $b_\pm$ do not intersect at the origin in the figure.
In \Fig{BoundedOrbits}(a) the line $b=0$ corresponds
to the one in \Fig{BoundedOrbitsNew} so that the same resonances are still relevant. These now extend
to the region $b>0$, bending strongly to the right.
The same overall resonance structure is also visible in \Fig{BoundedOrbits}(b).

%%%%%%%%%%%%%%%%%%%%%%%%%%%%%%%%%%%%%%%%%%%%%%%%%%%%%%%%%%%%%%%%%%%%%%%%%%%%%
\section{Coupled \Hen~Maps}\label{sec:HenonMaps}
%%%%%%%%%%%%%%%%%%%%%%%%%%%%%%%%%%%%%%%%%%%%%%%%%%%%%%%%%%%%%%%%%%%%%%%%%%%%%

Since Moser's map, \Eq{MoserMap} or equivalently \Eq{ShiftedMap}, is the generic, four-dimensional
quadratic map, there must be parameters for which it corresponds to a pair of uncoupled quadratic maps.
In this section, we show that this is possible for $\eps_2 = 1$ and special choices
of the matrix $C$, depending on $\eps_1$. The sign $\eps_1$ corresponds to positive and negative Krein signatures.
These \Hen~maps are uncoupled when $c= 0$, but when $c$ is nonzero the resulting coupling is,
as we will see in \Sec{Accelerator}, precisely what is needed to describe the dynamics in
the neighborhood of an accelerator mode of a \fourD{} standard map.
Whether there are other possibilities for which the
Moser map has uncoupled dynamics with respect to some invariant
canonical planes on which the dynamics is conjugate
to \Hen~maps is presently not clear and left for future study.

Dynamics of a pair of coupled \Hen~maps has been studied previously for example in \cite{Mao88, Ding90, Todesco94, Bountis94, Todesco96, Gemmi97, Vrahatis96, Vrahatis97, Giovannozzi98, Bountis06}, particularly in regard to models of storage rings for particle accelerators.

%%%%%%%%%%%%%%%%%
%% Decoupled Henon
%%%%%%%%%%%%%%%%%
\subsection{Decoupled Limits: \Hen~Maps}\label{sec:Decoupled}
To find parameter values for which the Moser map is decoupled, we search for a coordinate transformation that reveals the invariant planes.
This transformation should be affine in order to maintain the quadratic form. So that the resulting map is symplectic with the standard Poisson matrix \Eq{Symplectic} and to maintain the momentum-coordinate split for the \Hen~form, we start with the linear transformation:
$$
	(\xi,\eta) = S(q,p) = (A q, \rho A^{-T} p),
$$
where $A$ is invertible and $\rho > 0$.
This transformation is symplectic-with-multiplier, $DS^TJ DS = \rho J$. In the new coordinates
the map \Eq{ShiftedMap} becomes
\begin{align*}
	q' &=  q+ \hat{C}^{-T}(-p + \hat{C} q + \nabla_q \hat{U}(q)) ,\\
	p' &= \hat{C} q ,
\end{align*}
where
$$
	\hat{C} = \frac{1}{\rho} A^T C A, \qquad\qquad
	\hat{U}(q) = \frac{1}{\rho} U(Aq) ,
$$
so that $\nabla_q \hat{U}(q) = \tfrac{1}{\rho} A^T \nabla_\xi U(A q)$.
In order that the map be decoupled in the new coordinates, any cross terms in the new potential $\hat{U}$ should be zero.
This can be accomplished for \Eq{UPotential} only if $\eps_2  = 1$, $c = 0$, and $A$ is proportional to a rotation by angle $\pi/3$.
We can normalize the amplitude of the quadratic terms in the new map by setting $\rho =  1/\sqrt{12}$ choosing
$$
		A= \rho
				\begin{pmatrix} 1 & 1 \\ -\sqrt{3} & \sqrt{3} \end{pmatrix}
$$
to give
$$
	\hat{U}(q)  = (a - \sqrt{3}b)q_1 + (a + \sqrt{3}b)q_2 + \tfrac13(q_1^3 + q_2^3) .
$$
The transformation is thus fixed by this choice. In order that the resulting map be decoupled, $\hat{C}$ must be diagonal, e.g.,
\beq{Chat}
		\hat{C} = \begin{pmatrix} \eps_1 & 0 \\ 0 & 1 \end{pmatrix} ,
\eeq
where $\det{C} = \det{\hat{C}} = \eps_1 = \pm 1$. In order for this to be the case,
the original $C$ must take one of two forms:
\begin{subequations} \label{eq:CHenon}
\begin{align}
		\eps_1 = 1: \quad  C &= \begin{pmatrix} \sqrt{3} & 0 \\ 0 &\tfrac1{\sqrt{3}} \end{pmatrix},
		 \label{eq:CHenonPlus}\\
		\eps_1 = -1 : \quad C &= \begin{pmatrix} 0 & 1 \\ 1 &0 \end{pmatrix}. \label{eq:CHenonMinus}
\end{align}
\end{subequations}
Note that one could also replace $\hat{C}$ by $-\hat{C}$ in \Eq{Chat}, but by the symmetries discussed in \Sec{SecondDifference}, this gives nothing new.

With this we get the transformed map
\begin{align*}
	q'_1 &=  2q_1  + \eps_1(-p_1+ a - \sqrt{3}b + q_1^2) \\
	p'_1&= \eps_1 q_1 \\
	q'_2 &= 2q_2 -p_2 +  a + \sqrt{3}b+ q_2^2 \\
	p'_2&= q_2 .
\end{align*}
This is not quite in the \Hen~form \Eq{HenonMap}, but a final affine transformation $q \to (\hq_1-\eps_1,\hat{q_2}-1)$ and $p \to \hp - (1,1)$ brings the map to the form
\bsplit{decoupled-Henon-maps}
	\hq'_1 &=  \eps_1(-\hp_1+ 1+ a - \sqrt{3}b + \hq_1^2)\\
	\hp'_1&=  \eps_1\hq_1 \\
	\hq'_2 &= -\hp_2 +  1+a + \sqrt{3}b + \hq_2^2\\
	\hp'_2&= \hq_2 .
\esplit
Note that after this transformation we obtain, when $\eps_1 = 1$, a pair of uncoupled maps
of the \Hen-form \Eq{HenonMap}.
However, when $\eps_1=-1$, the first canonical pair has the \Hen~form only upon a mirroring
transformation, e.g., $(\hq_1,\hp_1) \to (-\hq_1,\hp_1)$; the point is that in the canonical coordinates, the $(\hq_1,\hp_1)$ components  ``rotate" under the map in the opposite sense from $(\hq_2,\hp_2)$.

The component maps have saddle-center bifurcations along the lines $a = \pm\sqrt{3}b$,
creating pairs of fixed points
for each component when $a < \pm \sqrt{3}b$, respectively.
However, in order that the  \fourD{} map have a fixed point, both components must
have fixed points, implying that $a < -\sqrt{3}|b|$. This bifurcation occurs on the same line that appear in \Fig{BifSurfaces}(a) when $c = 0$.

When the maps are decoupled, and $a < -\sqrt{3}|b|$ the four newly created fixed points
have types EE, EH, EH and HH. In particular the EE point is located at
\begin{align*}
	\eps_1q_1^* = p_1^* &= 1 - \sqrt{\sqrt{3}b-a} ,\\
	q_2^* = p_2^*       &= 1-\sqrt{-\sqrt{3}b-a} .
\end{align*}
In the original variables, the
fixed points of these maps are given by \Eq{cZeroFPs} since this corresponds to $c=0$.
The doubly elliptic fixed point remains stable in the rectangle
$$
	-4+\sqrt{3}|b| < a < -\sqrt{3}|b|
$$
since the individual maps have period-doubling bifurcations at $a = -4\pm \sqrt{3}b$, respectively.

Recall that for a symplectic map \Eq{Symplectic}, with a doubly elliptic fixed point $z^*$, the quadratic form $q(v) = v^T JDf(z^*) v$ is an invariant of the linearized dynamics. This implies stability when the form $q$ is definite \cite{Arnold68,Howard87}: an EE fixed point cannot cross the $KP=0$ line in \Fig{ABPlane} into the CU region. Equivalently, the Krein bifurcation cannot occur if the symmetric matrix
$$
	{\cal Q} = \tfrac12 [JDf(z^*) - Df^T(z^*) J]
$$
is definite. At an elliptic-elliptic point for \Eq{decoupled-Henon-maps}, ${\cal Q}$
becomes
$$
	{\cal Q} = \left(\begin{array}{c|c} \hat{C} &
						\begin{array}{cc} -\eps_1 q_1^*   & 0 \\ 0 &-q_2^*\end{array} \\
	 								\hline
						\begin{array}{cc} -\eps_1 q_1^*   & 0 \\ 0 &-q_2^*\end{array} & \hat{C} \end{array}\right) ,
$$
where $\hat{C}$ is given in \Eq{Chat}.
When \Eq{decoupled-Henon-maps} has an EE point then $q_i^{*2} < 1$. This implies that ${\cal Q}$ is positive definite when $\eps_1 = 1$, but has a pair of negative eigenvalues when $\eps_1 = -1$. Thus, only in the latter case, can coupling lead to a Krein bifurcation.

Indeed, we can see this if we re-introduce coupling by allowing $c \neq 0$.
The same transformations  that lead to \Eq{decoupled-Henon-maps} can be applied if we still take $\eps_2 = 1$ and $C$ to be one of the matrices \Eq{CHenon}. The result is the pair of coupled \Hen~maps
\bsplit{decoupled-Henon-maps-with-coupling}
	\hq'_1 &=  \eps_1\left(-\hp_1+ \ahenOne  + \hq_1^2 + \frac{c}{2\sqrt{3}}(\hq_1 + \hq_2)\right) ,\\
	\hp'_1&=  \eps_1q_1 ,\\
	\hq'_2 &= -\hp_2 + \ahenTwo + \hq_2^2 + \frac{c}{2\sqrt{3}}(\hq_1 + \hq_2) ,\\
	\hp'_2&= \hq_2 ,
\esplit
where
\bsplit{coupled-Henon-params}
	\ahenOne &= 1+ a - \sqrt{3}b -\frac{\eps_1+1}{2\sqrt{3}}c ,\\
	\ahenTwo &= 1+a + \sqrt{3}b  - \frac{\eps_1+1}{2\sqrt{3}}c .
\esplit

The stability of the fixed points for this map
can be conveniently obtained from the results
for the Moser map \Eq{ShiftedMap}.
For $\eps_1 = \eps_2 = 1$ and the diagonal matrix in \Eq{CHenonPlus},
the stability parameters \Eq{AB} become
\begin{align*}
	A &= 4(1+\sqrt{3}\xi_1) +\frac{c}{\sqrt{3}} ,\\
	B &= 6 + \frac{2c}{\sqrt{3}}(1+\sqrt{3}\xi_1) +4\sqrt{3}\xi_1(2+ \sqrt{3}\xi_1)   - 4\xi_2^2 .
\end{align*}
Since the quadfurcation occurs for $a=b=c=0$,  $\xi_1=\xi_2 = 0$, it always occurs at the point
$(A^\txtQ,B^\txtQ) = (4,6)$ for this $C$, as more generally true for the symmetric case in \Sec{Symmetric}. For this case the Krein parameter
\Eq{Krein} is always nonpositive:
\beq{KreinHenonPlus}
	KP = -4\xi_2^2 - c^2/12,
\eeq
even when the maps are coupled again by nonzero $c$.
This reflects the fact that for the two \twoD{} uncoupled
maps the elliptic motions in $(\hp_1, \hq_1)$ and $(\hp_2, \hq_2)$
have the the same orientation. This is not changed
by the coupling when $\eps_1=+1$, and a CU instability
is not possible. This case corresponds to the transition \Eq{EEHHTwoEH} with
a stability diagram like that shown in \Fig{ABDiagrams2}(b)
and neither the EE nor the HH fixed point may turn CU.

For the case $\eps_2=1$, and $\eps_1 = -1$,
for the matrix in \Eq{CHenonMinus}, we similarly obtain
\begin{align*}
	A &= 4(1+\xi_2) ,\\
	B &= 6 -12\xi_1^2 + 8\xi_2 + 4\xi_2^2-2c\xi_1 .
\end{align*}
Now the Krein parameter becomes
\beq{KreinHenonMinus}
	KP = -2\xi_1(6\xi_1 + c),
\eeq
which potentially may have either sign. So  when
the rotation directions for the two canonical planes are opposed, the coupling terms
make a Krein bifurcation possible as $c$ crosses zero.
Thus the initial quadfurcation may correspond to
the transition  $\emptyset \to$ EE $+$ HH $+$ 2 EH, as in  \Fig{ABDiagrams2}(b),
and both the EE and HH might later become CU.
In addition a direct transition to 2 CU $+$ 2 EH
is possible, which is analogous to the $A$-$B$ diagrams shown
in \Fig{ABDiagrams2}(a) for the fully coupled case.

%%%%%%%%%%%%%%%%%
%% Henon Numerical experiments
%%%%%%%%%%%%%%%%%
\subsection{ Numerical illustration}\label{sec:CoupledHenon}

%%%%%
\InsertFig[b]{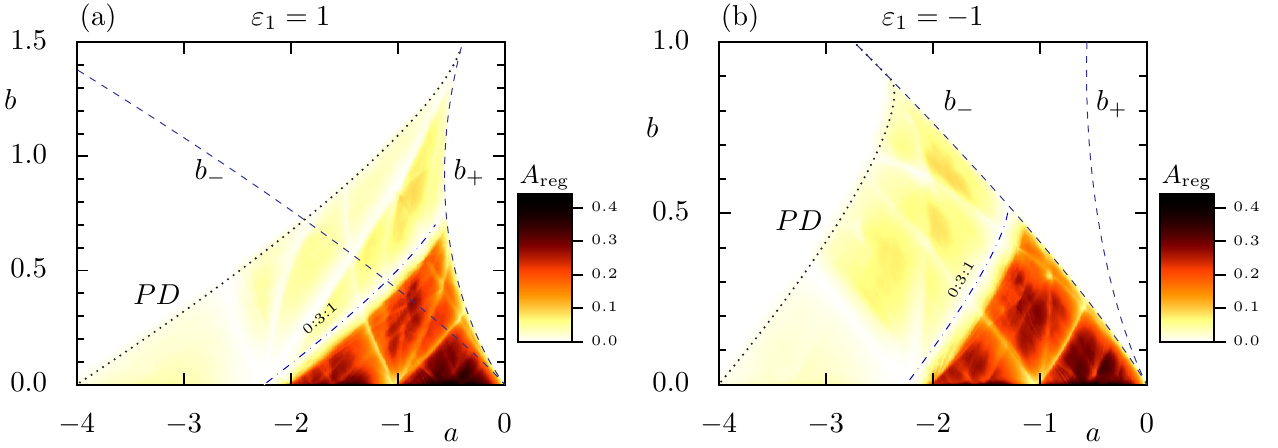}
{Area $\Areg$ of bounded orbits on the plane $\eta = C\xi$ for
the matrices \Eq{CHenon} with $\eps_2 = 1$, $c = 2b$ and for (a) $\eps_1 = 1$, and (b) $\eps_1 = -1$. The maps reduce to the uncoupled case only along the axis $b=0$, since then $c=0$ as well.
The curves $b_\pm$ represent the crossing of the
surfaces \Eq{bPlusMinus} with the plane $c = 2b$.
Thus there are four fixed points when
$b$ is below $b_-$ and two when $b$ is between $b_-$ and $b_+$.
For both examples,
to the left of the (dotted) $PD$ line the EE fixed point becomes EI.
The dashed-dotted line shows the $0:3:1$
resonance near where there is a strong decrease in $\Areg$.
}
{HenonBounded}{WidthIsDeterminedFromPdfFile}
%%%%%

\begin{figure}[t]
\includegraphics[scale=1.2]{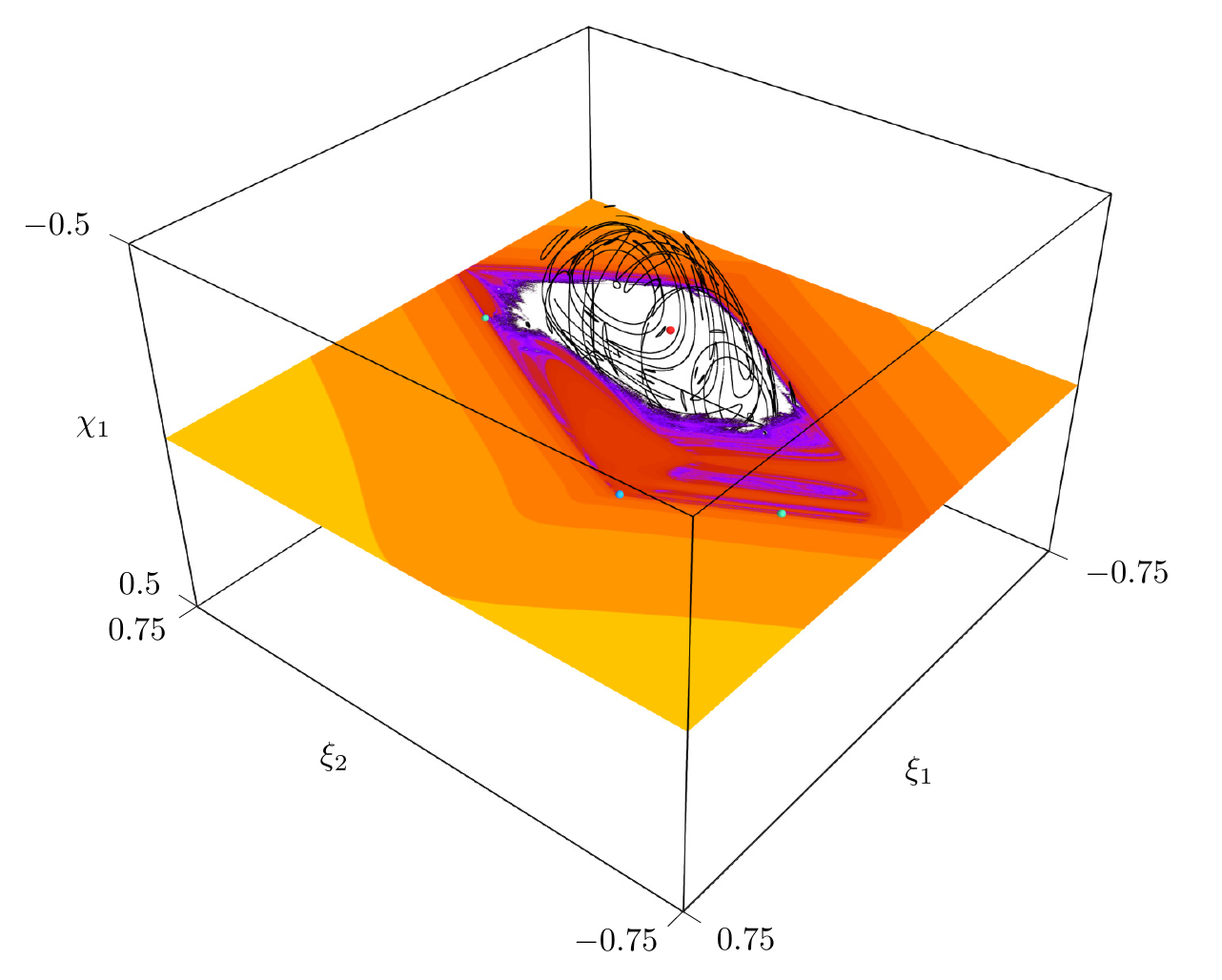}
\caption{3D phase space slice
         and corresponding escape time plot in the $\eta = C\xi$ plane
         for the diagonal matrix $C$ of \Eq{CHenonPlus} with $\eps_1 = 1$
         and  $(a, b, c) = (-0.3, 0.1, 0.2)$.
         Shown are several selected regular tori (black lines)
         near the ``outer edge'' of the ``regular region''
         surrounding the EE fixed point.
         Each torus is represented by $10^4$ points in the
         slice with $\epsilon=10^{-6}$.
         The four fixed points are shown as small spheres:
         EE (red), 2 EH (green), and HH (blue).
         The coloring of the escape times is the same
         as in \Fig{EscapeTimePlot}.
         \movierefall
          }
         \label{fig:uncouppled-coupled-Henon}
\end{figure}

Let us now illustrate the dynamics near the uncoupled case.
Figure~\ref{fig:HenonBounded} shows
the area, $\Areg$, of bounded orbits
for the two cases \Eq{CHenon} of the matrix $C$.
In these figures, the maps are uncoupled along the line $b = 0$
since we choose $c = 2b$.
As in Figs.~\ref{fig:BoundedOrbitsNew} and
\ref{fig:BoundedOrbits} one observes clear drops
of $\Areg$ along curves in the $(a, b)$ plane.
When $b=0$ both plots in \Fig{HenonBounded}
agree so that also the association
to the relevant resonances  is the same.
For both panels, the $b_-$ line corresponds
to a saddle-center bifurcation, but in
\Fig{HenonBounded}(a) two fixed points with stabilities
EH and HH disappear upon reaching $b_-$ from
below; thus the decrease in the number of fixed points
has no significant influence on the size of $\Areg$.
In contrast, for \Fig{HenonBounded}(b) the fixed
points with stabilities EE and EH disappear upon reaching $b_-$ from
below, so that only the EH and HH fixed points are left
in the region between $b_{-}$ and $b_{+}$
and $\Areg$ decreases abruptly.
In both cases the $0:3:1$ resonance
of the EE fixed point leads to a strong reduction in bounded area
as $a$ decreases through the resonance.

As $\eps_1=1$ for \Fig{HenonBounded}(a) none of fixed points can become CU.
For $\eps_1=-1$, as in \Fig{HenonBounded}(b),
the HH fixed point becomes CU when $a$ is sufficiently negative, but
again this does not significantly influence $\Areg$.

Figure~\ref{fig:uncouppled-coupled-Henon}
shows a \threeD{} phase space slice plot
and an escape time plot in the $(\xi_1, \xi_2)$-plane
for $(a, b, c) = (-0.3, 0.1, 0.2)$ and $\eps_1=1$.
The elliptic-elliptic fixed point is surrounded by
a region of predominantly regular motion
as seen by the white region of non-escaping
orbits (within $10^4$ iterations).
Corresponding regular \twoD{} tori are shown as black
curves in the slice. Some of these are secondary tori around periodic orbits and appear
as sequences of disjoint loops in the \threeD{} phase space slice.
The HH fixed point and the two EH fixed points
approximately limit the region in which the regular orbits
and orbits with longer escape times are contained.

%%%%%%%%%%%%%%%%%
%% Accelerator Modes
%%%%%%%%%%%%%%%%%
\section{Accelerator Mode Islands}\label{sec:Accelerator}

Accelerator modes are special orbits that occur in action-angle dynamical systems that are periodic in action as well in angle.
For example,  Chirikov's area-preserving map \cite{Chi1979} on $(p,q) \in \bR \times \bT$ can be written
\bsplit{map2d}
      p' &= p + \tfrac{K}{2\pi} \sin(2\pi q') , \\
      q' &= q + \eps p            \mod 1  .
\esplit
Here we have added an additional parameter, $\eps = \pm 1$, representing the direction of the twist; this will be used for the \fourD{} case below.
For this system, projecting the ``action" or momentum, $p$, onto the torus gives a smooth map on $(p,q) \in \bT^2$.
An accelerator mode is a period-$T$ orbit of this projected map
that lifts to an orbit on $\bR \times \bT$ for which
\bsplit{AccelMode}
	p_T &= p_0 + m, \\
	q_T  &=q_0 \mod 1
\esplit
where $m \in \bZ \setminus 0$, i.e., the momentum increases by $m$ each period, so the orbit accelerates.
Such orbits were first studied  in \cite{Zaslavsky70, ChiIzr1973}.
Regular islands surrounding an elliptic accelerator mode can have a substantial impact
on the broader dynamics of the map. In particular,
chaotic trajectories for the lifted map may show super-diffusive behavior
in momentum due to long-time stickiness near the regular islands, see e.g.,
\cite{ Chi1979, Karney82, ZumKla1994, ZasEdeNiy1997, RomZas1999, Manos14, MigSimVie2015, AluFisMei2017}.

In this section we consider the \fourD{} symplectic map defined by
\bsplit{Froeschle}
	p_1' &= p_1 + \tfrac{1}{2\pi}\left[ K_1 \sin(2\pi q_1') + L \sin(2\pi(q_1'+q_2'))\right] ,\\
	p_2' &= p_2 + \tfrac{1}{2\pi}\left[ K_2 \sin(2\pi q_2') + L \sin(2\pi(q_1'+q_2)')\right] ,\\
	q_1' &= q_1 + \eps_1 p_1 ,\\
	q_2' &= q_2 + p_2 ,
\esplit
where $(p,q) \in \bR^2 \times \bT^2$. For the case $\eps_1 = 1$, this is the map first studied by Froeschl\'e \cite{Fro1971,Fro1972}, and when $\eps_1 = -1$, it has indefinite twist, and is equivalent to the map studied by Pfenniger \cite{Pfenniger85a}.  As we will see below, the parameter $\eps_1$ is related to the same parameter of the Moser map \Eq{MoserMap}.

Periodicity in the action variables implies that \Eq{Froeschle} can be projected onto the four-torus $\bT^2 \times \bT^2$,
and---as before---a period-$T$ orbit obeying \Eq{AccelMode}, now with $m \in \bZ^2\setminus (0,0)$,
is an accelerator mode. These were first studied in \cite{KooMei1990} where it was shown that like their \twoD{} counterparts, they can have a substantial effect on the action-diffusion in strongly chaotic regimes.

For the \twoD{} case, it was first shown in \cite{Karney82} that near the saddle-center bifurcation that creates
an accelerator mode of \Eq{map2d}, the dynamics can be approximated by the \Hen~map \Eq{HenonMap}, we recall
this derivation in \Sec{2DAccel} below.
Since the map \eqref{eq:Froeschle} consists of two \twoD{}
standard maps that are coupled when $\hFro \neq 0$, it is perhaps not surprising that a similar result
holds for \Eq{Froeschle} near the creation of an accelerator mode, see \Sec{4DAccel}. When $\hFro \ll 1$,
we will relate the local dynamics to the pair of coupled \Hen~maps \Eq{decoupled-Henon-maps-with-coupling}, and hence to the quadfurcation in the Moser map \Eq{ShiftedMap}.

%----------------------------------------------------------------------------
\subsection{Accelerator Modes for the \twoD{} Standard Map} \label{sec:2DAccel}
%----------------------------------------------------------------------------

Let us first recall the relation between the dynamics in the neighborhood of
a fixed point accelerator mode for \Eq{map2d} and the \Hen~map.
Since the form \Eq{HenonMap} is different from the one used in \cite{Karney82},
we use a different expansion and coordinate change here.

A fixed point of \Eq{map2d} on the torus must satisfy
$q'=q \mod 1$ and  $p'=p \mod 1$.
From \Eq{map2d} it follows that $p = 0 \mod 1$ so we can take $p = p^\star = 0$, and that
$q$ must solve
\begin{equation}
  2\pi m = K \sin(2\pi q)
\end{equation}
with $m\in \bZ$, accounting for the periodic boundary condition in $p$.
For $m=0$ one gets the fixed
points $(0, 0)$ and $(0, \tfrac12)$, which exist for all values of $K$.
For $m\neq 0$ one gets accelerator modes; these exist only
for sufficiently large $|K|$. Restricting to $m = 1$, and taking $K\ge 0$ gives two solutions
\begin{align*}
   q^\star_\times &= \tfrac{1}{2\pi} \arcsin \left( \tfrac{2\pi }{K} \right) ,\\
   q^\star_\circ &= \tfrac{1}{2}
                - \tfrac{1}{2\pi} \arcsin \left( \tfrac{2\pi }{K} \right) ,
\end{align*}
provided that $K > 2\pi$.
This pair of fixed points is created at $K = K^*=2\pi $ in a
saddle-center bifurcation at $(0,\tfrac14)$.
The fixed point at $(0,q^\star_\times)$
is hyperbolic when $K > 2\pi$ and that at $(0,q^\star_\circ)$
is elliptic when $2\pi < K < 2\pi\sqrt{1+4/\pi^2}$.

To relate the local dynamics to the \twoD{}  \Hen~map, we transform to coordinates
centered at the bifurcation,
$$
	(\tp,\tq) = (p,q) -  \left(0,\tfrac14 \right), \quad \tK = \pi(K - 2\pi),
$$
and assume that there is a formal parameter $\Delta \ll 1$ so that
$\tq,\tp, \tK = \cO(\Delta)$.
Expanding the map \Eq{map2d}
in the new coordinates gives, through second order in $\Delta$,
\begin{align*}
      \tp'  & = \tp + \tfrac{\tK}{2\pi^2}
                   -2\pi^2 \tq'^2 \\
      \tq' &= \tq + \eps \tp .
\end{align*}
Since this is a quadratic area-preserving map, it can be converted by an affine transformation
to the \Hen~form \Eq{HenonMap}; here we can use the transformation
\bsplit{HTransform}
	\hp &=1-2\pi^2 \eps \tq ,\\
	\hq &= \eps-2\pi^2(\tq +\eps\tp),
\esplit
to obtain, in the new coordinates,
\bsplit{henon-final}
      \hq' &=  \eps( 1-\tK  - \hp + \hq^2)\\
      \hp' &= \eps \hq  .
\esplit
When $\eps = 1$, this is the form \Eq{HenonMap} with $\ahenon = 1-\tK$.
When $\eps = -1$ the \Hen~map has the opposite rotation direction.
The fixed points of the map \Eq{henon-final}
occur when $\tK\ge0$ at
$\hp^\star = \eps \hq^\star = 1 \pm \sqrt{1-\ahenon} = 1 \pm \sqrt{\tK}$
and correspond to hyperbolic ($+$) and (initially)
elliptic ($-$) stability.
At $\ahenon = -3$ the elliptic fixed point
undergoes a period-doubling bifurcation and becomes inverse hyperbolic. This corresponds to $\tK=4$.
This value is close to the actual period-doubling of \Eq{map2d}
at $K = 2\pi\sqrt{1+4/\pi^2} \approx 7.448 $, which gives $\tK \approx 3.661 $.

%----------------------------------------------------------------------------
\subsection{Accelerator Modes for the \fourD{} Standard Map}\label{sec:4DAccel}
%----------------------------------------------------------------------------

A fixed point of the map \eqref{eq:Froeschle} must have $p_1 = p_2 = 0$ (mod 1) and
thus the coordinates of fixed point accelerator modes must be solutions of
\begin{equation} \label{eq:FroeschleAccel}
      \begin{aligned}
          2\pi m_1 &= K_1 \sin(2\pi q_1) + \hFro \sin(2\pi(q_1 + q_2)) , \\
          2\pi m_2 &= K_2 \sin(2\pi q_2) + \hFro \sin(2\pi(q_1 + q_2)) ,
      \end{aligned}
\end{equation}
for $m \in \bZ^2 \backslash (0, 0)$.
When $\hFro=0$, \Eq{Froeschle} reduces to a pair of uncoupled \twoD{} standard maps.
To study the behavior in the neighborhood
of an accelerator mode, we consider the bifurcation that occurs at
$(K_1,K_2,\hFro) = (2\pi,2\pi,0)$, where $q_1 = q_2 = \tfrac14$,
creating an $m = (1,1)$, fixed-point accelerator mode.

Whenever \Eq{Froeschle} has a fixed point we can evaluate its stability using
the parameters $A$ and $B$ in \Eq{ReducedPolyn}. This leads to the saddle-center parameter \Eq{BifCuvres}
\bsplit{SCFroeschle}
   SC = \eps_1 &\left[ K_1 \cos(2\pi q_1) + K_2 \cos(2\pi q_2)\right]
          L \cos\left (2\pi(q_1 + q_2)\right) \\
            & + \eps_1 K_1 K_2 \cos(2\pi q_1) \cos(2\pi q_2) .
\esplit
Note that when $q_1 = q_2 = \tfrac14$, then $SC =0$, so the accelerator modes are born on the $SC$ line.
The Krein parameter \Eq{KreinBifs} becomes
\beq{KreinFroeschleP}
	KP_{\eps_1 = 1}=  -\tfrac14 \left[K_1 \cos(2 \pi q_1) -K_2 \cos(2 \pi q_2)\right]^2
			 - [L \cos(2\pi(q_{1} + q_{2}))]^2,
\eeq
or
\bsplit{KreinFroeschleM}
	KP_{\eps_1 = -1} =   -\tfrac14  &\left[ K_1 \cos(2 \pi q_1 ) + K_2 \cos(2\pi q_2)\right] \times \\
 				      &\left[ K_1 \cos(2\pi q_1) + K_2 \cos(2\pi q_2)
						+ 4L \cos\left(2\pi(q_1 + q_2)\right) \right],
\esplit
depending upon $\eps_1$.
Consequently, when $\eps_1 = 1$, $KP \le 0$ and there are no CU fixed points (recall \cite{Kook89}).
However, when $\eps_1 = -1$, then $KP$ can have either sign. In both cases the $m = (1,1)$
accelerator mode born at $L=0$ has $KP=0$. Thus these fixed points are born at $(A,B) = (4,6)$ with a quadruple unit eigenvalue.

To find a quadratic approximation in the neighborhood of the accelerator mode,
we have to perform similar transformations to the \twoD{} case, namely we expand using
$$
   (\tp_i, \tq_i) = (p_i,q_i) - (0, \tfrac14), \quad
   \tK_i  = \pi (K_i-2\pi) .
$$
As before we assume that the coordinate and parameter deviations are all of the same order of
smallness, $\tq_i$, $\tp_i$,  $\tK_i$, $\hFro = \cO(\Delta)$.
Expanding \Eq{Froeschle} then gives, to $\cO(\Delta^2)$,
\begin{align*}
      \tp_1' &= \tp_1 + \frac{\tK_1}{2\pi^2}  -2\pi^2\tq_1'^2
                - \hFro (\tq_1' + \tq_2') ,\\
      \tp_2' &= \tp_2 + \frac{\tK_2}{2\pi^2}  -2\pi^2\tq_2'^2
                - \hFro (\tq_1' + \tq_2') ,\\
      \tq_1' &= \tq_1 + \eps_1\tp_1 ,\\
      \tq_2' &= \tq_2 + \tp_2 .
\end{align*}
Slightly generalizing the transformation \Eq{HTransform}, we let
\bsplit{HTransform4D}
	\hp_1 &=1-2\pi^2 \eps_1\tq_1 ,\\
	\hq_1 &= \eps_1-2\pi^2(\tq_1 + \eps_1\tp_1),\\
	\hp_2 &=1-2\pi^2 \tq_2,\\
	\hq_2 &= 1-2\pi^2(\tq_2+\tp_2) .
\esplit
This transformation then gives
the pair of coupled \Hen~maps \Eq{decoupled-Henon-maps-with-coupling} if we
identify
\begin{align*}
	\ahenOne &= 1-\tK_1 +(\eps_1+1) L, \\
	\ahenTwo &= 1-\tK_2 +(\eps_1+1) L, \\
       c &= -2\sqrt{3} \eps_1 \hFro .
\end{align*}
Thus the neighborhood of a fixed point accelerator mode for the Froeschl\'e map is described by the Moser map \Eq{ShiftedMap} with $\eps_2 = 1$, one of the matrices $C$ in \Eq{CHenon} and by \Eq{coupled-Henon-params}
\begin{align*}
   a &= - \tfrac12 (\tK_1 + \tK_2)
      = \tfrac{\pi}{2} (4\pi -K_1 - K_2), \\
   b &= \tfrac{1}{2\sqrt{3}} (\tK_1 - \tK_2)
      = \tfrac{\pi}{2\sqrt{3}} (K_1 - K_2) ,\\
   c &= -2\sqrt{3} \eps_1 \hFro .
\end{align*}
The implication is that the creation of an accelerator mode for \Eq{Froeschle}
at $K_1 = K_2 = 2\pi$ and $\hFro = 0$
is locally described by a special case of the quadfurcation of the Moser map.

Note that the form of the coupled \Hen~maps changes for $\eps_1 = \pm 1$, as indicated in \Eq{decoupled-Henon-maps-with-coupling}.
Transforming back to the Moser coordinates shows that the case $\eps_1 = +1$  corresponds to the diagonal matrix  \Eq{CHenonPlus}, and $\eps_1=-1$ to the anti-diagonal matrix \Eq{CHenonMinus}.
When $\eps_1 = 1$, the Krein parameter is given by \Eq{KreinHenonPlus}, which is non-positive in agreement with \Eq{KreinFroeschleP}, so fixed points of type CU are not possible.
Since the matrix \Eq{CHenonPlus} is symmetric,
the creation of the accelerator mode follows
the pattern shown in \Fig{ABDiagrams2}(b)
and the EE and HH branches above the $SC$ line
must stay outside of the CU region.
When $\eps_1=-1$, however, then there can be a direct bifurcation to
CU fixed points since the Krein parameter \Eq{KreinHenonMinus}
does not have a definite sign, in agreement with \Eq{KreinFroeschleM}.
Thus the bifurcation may either follow the pattern
as shown in \Fig{ABDiagrams2}(a) with a direct transition
to CU pair, or the one shown in in \Fig{ABDiagrams2}(b),
but with the possibility that the EE or HH (or both)
fixed point may eventually become CU.

%%%%%%%%%%%%%%%
%----------------------------------------------------------------------------
\subsection{Numerical Illustration}
%----------------------------------------------------------------------------

As an example, we consider the case $\eps_1 = 1$. When the
coupling $\hFro=0$, the saddle-center bifurcations of the \twoD{} maps in
$(q_i, p_i)$ occur at $K_i=2\pi$
and correspond to a quadfurcation of the form \Eq{EEHHTwoEH}, i.e.,
$$
	\emptyset \to \mbox{EE} + \mbox{2 EH} + \mbox{HH} .
$$
When $\hFro \neq 0$, these fixed points persist because when $K_i > 2\pi$ they do not have a unit eigenvalue;
moreover when  $\hFro \ll 1$, the four fixed points will have the same stability types
since the eigenvalues of the linearization are
continuous functions of the coupling, and as noted above, when $\eps_1 = 1$ there is no Krein bifurcation.
A \threeD{} slice in the phase space for $K_1=6.4$, $K_2=6.5$, and $\hFro=0.05$
is shown in \Fig{standard-map-4d-3DPSS}.
The EE fixed point is surrounded by families of two-tori
and some selected examples are shown; as in previous plots each torus typically intersects the
slice in two loops.
%%
%%%%%
\InsertFig[b]{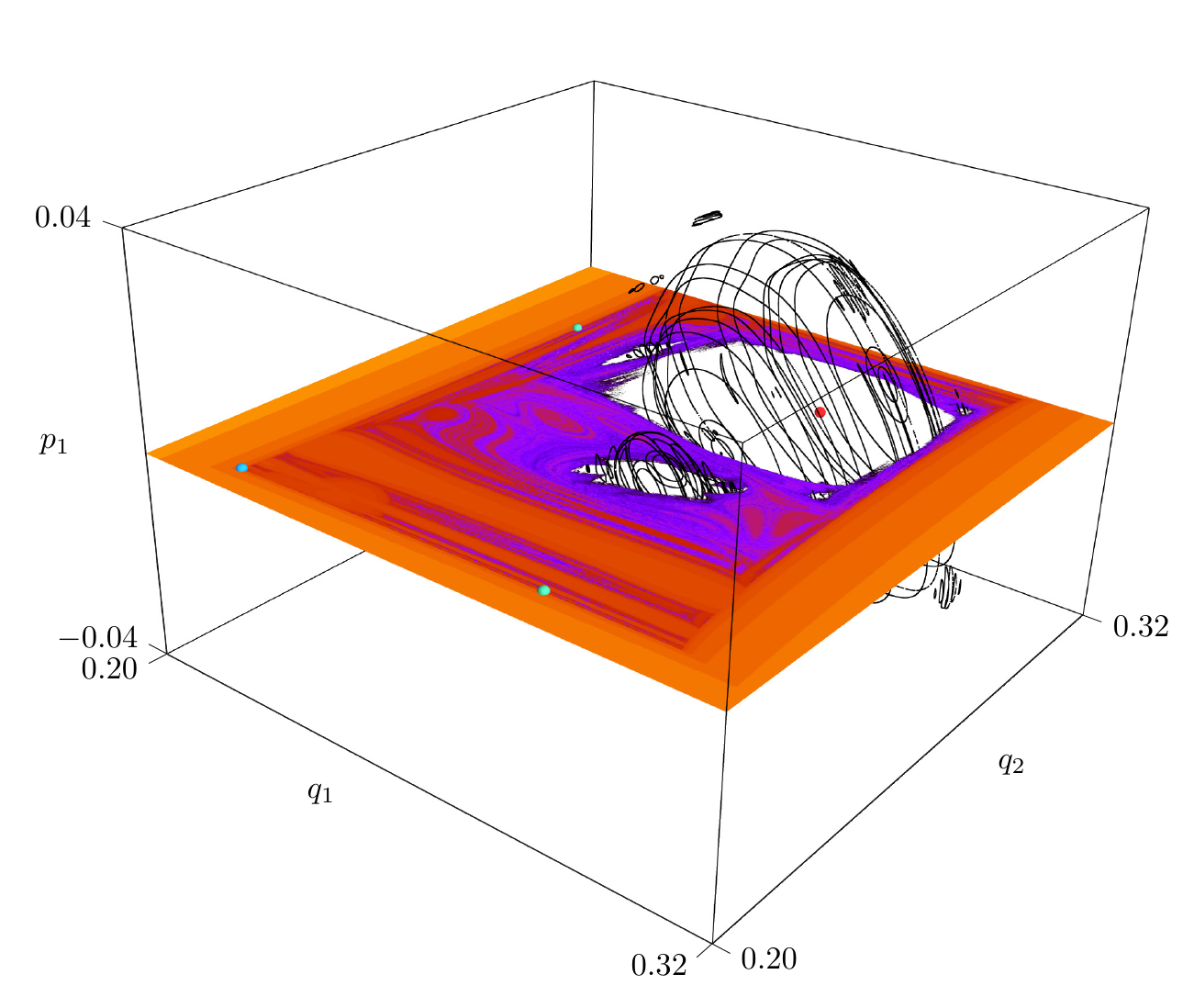}
{Geometry of accelerator modes in the \fourD{} coupled standard map
for $\eps_1 = 1$, $K_1=6.4$, $K_2=6.5$ and $\hFro=0.05$.
Shown is a \threeD{} phase space slice $|p_2| < \eps=10^{-6}$
and a corresponding escape time plot in the $p_1=p_2 = 0$ plane.
The EE fixed point (red sphere) is surrounded by
families of regular \twoD{} tori (black curves),
each shown with $10^5$ points in the slice.
The coloring of the escape times is the same as in \Fig{EscapeTimePlot}.
\movierefall}
{standard-map-4d-3DPSS}{WidthIsDeterminedFromPdfFile}
%%%%%
%%
A corresponding escape time plot in the $(q_1, q_2)$-plane
for initial conditions with $(p_1, p_2) = (0, 0)$
is shown as a plane in the figure with escape time encoded
in color as in the previous plots.
The white region of long-time (here $t_{\text{max}} = 10^4$) trapped orbits
mainly surrounds the EE fixed point
and corresponds closely to the region containing the regular tori.
Note, however, that initial conditions near the
boundary of this region may eventually escape after some longer
transient time.
The location of the regular region can be
understood from the case of the uncoupled H\'enon maps:
for both $q_1$ and $q_2$
there is an interval of initial conditions with regular motion
that is limited by the chaotic dynamics caused by homoclinic intersections
of the stable and unstable manifolds of
the hyperbolic fixed points. The direct product structure
of these intervals is still reflected in the weakly coupled case.
Beyond the white region surrounding the EE fixed point
there are several smaller ones. These are
mainly due to regular tori fulfilling resonance conditions;
for example, the largest of these smaller regions,
located towards the EH fixed point at $(q_1, q_2) \approx (0.28 0.21)$,
is due to a $5:0:1$ resonance.
The next smaller white region, located towards the other EH
fixed point at $(q_1, q_2) = (0.22, 0.29)$, is
caused by a $7:0:1$ resonance.

%%%%%%%%%%%%%%%%%%%%%%%%%%%%%%%%%%%%%%%%%%%%%%%%%%%%%%%%%%%%%%%%%%%%%%%%%%%%%
\section{Summary and Outlook}
%%%%%%%%%%%%%%%%%%%%%%%%%%%%%%%%%%%%%%%%%%%%%%%%%%%%%%%%%%%%%%%%%%%%%%%%%%%%%
In this paper we have studied some of the dynamics of Moser's \fourD{}
quadratic, symplectic maps, which have a normal form \Eq{ShiftedMap} with six parameters $a,b,c,\alpha,\delta,\mu$ and two discrete parameters $\eps_1,\eps_2$.
We showed that there is a codimension-three submanifold in parameter space for which
Moser's map has a single fixed point with a pair of unit eigenvalues.
Bifurcations that occur on this submanifold correspond to creation and destruction of up to four fixed points,
the maximum possible number for the map (except for one singular case).
Along paths in parameter space that pass through this singularity it is possible that
four fixed points are created from none---a quadfurcation.
For other paths, two fixed points may be created or collide and emerge as two or four,
recall \Fig{BifSurfaces} and  \Tbl{NumFixed}.
Intuitively, the simplest case corresponds to a pair of uncoupled
\twoD{} \Hen~maps, where a quadfurcation corresponds
to choosing parameters so that the maps have simultaneous, co-located saddle-center bifurcations.

When a fixed point has four distinct eigenvalues on the unit circle (has type ``EE"), then it is linearly stable,
and according to KAM theory, is generically surrounded by a Cantor family of invariant two-tori.
We have seen that it is possible for one or two EE-fixed points to emerge from the quadfurcation.
The remaining two fixed points have at least one hyperbolic pair of eigenvalues. For the case of uncoupled
\Hen~maps, the four fixed points correspond to the cross-products of the saddles and centers of the two area-preserving maps,
giving rise to a single EE fixed point, two EH points, and one HH point. This scenario persists when coupling is added, and
describes the creation of accelerator modes in a \fourD{} standard map near zero coupling.
However, this scenario is rather special from the point of Moser's map,
where one more typically has the creation of two EE and two EH fixed points,
unless the matrix $C$ is symmetric.

For symmetric $C$, where the map is reversible, the quadfurcation has the special feature that the fixed points
are born with four unit eigenvalues. This allows, for example, the direct creation of complex unstable, CU,
fixed points. It is also of interest that when $C$ is symmetric, the limiting form of Moser's map near the
quadfurcation is a natural Hamiltonian system \Eq{HamSymmetric}. It is still an open question whether the
map is reversible \textit{only} when $C$ is symmetric

We showed in \Th{Bounded} that, when $\eps_2 \neq 0$, there is a ball that contains all bounded orbits.
We observe that orbits that remain bounded are typically associated with the EE fixed points. The computations suggest that the
center-stable manifolds of the EH points are likely candidates for partial barriers that delineate the boundary of
the region of orbits that have long escape times. In a future paper, we hope to compute these manifolds
to understand better the geometry of the barriers.

There are several other interesting questions left for future studies.
For the \twoD{} case, where the \Hen{} map
provides the universal form for any quadratic area-preserving map,
the algebraic decay of the survival probability
for the escape from a neighborhood of the regular region
is well established and understood in terms of partial barriers
and approximately described by Markov models. While for higher-dimensional
maps such power-law stickiness is numerically well established,
the mechanism for this is still not understood.
Moser's map is the prototypical example
for the study of the stickiness of a regular region in \fourD.
Of course, in this context, Arnold diffusion will also be important.

An interesting related aspect
is the study of the accelerator mode islands
in \fourD{}, where the dynamics is much richer
than the \twoD{} case,
due to the varied classes of stability and resonances.
We leave for future study the form of the local dynamics near accelerator modes
with  $m_1 \neq m_2$, as well as of modes born when $\hFro \neq 0$, in \Eq{FroeschleAccel}.

Finally, it would be of interest to study similar bifurcations for polynomial maps of higher degree; for example, the cubic case can be written in the Moser form as a composition of affine maps with a symplectic shear \cite{Koch14}. And of course, one can wonder what other exotic local bifurcations may happen in even higher-dimensional symplectic maps.

%%%%%%%%%%%%%%%%%%%%%%%%%%%%%%%%%%%%%%%%%%%%%%%%%%%%%%%%%%%%%%%%%%%%%%%%%%%%%
\section*{Acknowledgements}
%%%%%%%%%%%%%%%%%%%%%%%%%%%%%%%%%%%%%%%%%%%%%%%%%%%%%%%%%%%%%%%%%%%%%%%%%%%%%

We would like to thank Robert Easton, Roland Ketzmerick, Rafael de la Llave,
and Martin Richter for useful discussions.
JDM was supported in part by NSF grant DMS-1211350, and as a
Dresden Senior Fellow of the
Technische Universit\"at Dresden.
AB acknowledges support by the Deutsche Forschungsgemeinschaft under grant KE~537/6--1.
The visualizations of the \threeD{} phase space slices were created using
\textsc{Mayavi}~\cite{RamVar2011}.

%%%%%%%%%%%%%%%%%%%%%%%%%%%%%%%%%%%%%%%%%%%%%%%%%%%%%%%%%%%%%%%%%%%%%%%%%%%%%
\section*{References}
%%%%%%%%%%%%%%%%%%%%%%%%%%%%%%%%%%%%%%%%%%%%%%%%%%%%%%%%%%%%%%%%%%%%%%%%%%%%%

\providecommand{\newblock}{}

\end{document}